\documentclass[11pt]{article}
\usepackage[margin=1in]{geometry}
\usepackage[utf8]{inputenc}
\usepackage[T1]{fontenc}
\usepackage{microtype}
\usepackage{amsmath}
\usepackage{amssymb}
\usepackage{amsthm}
\usepackage{mathtools}
\usepackage{dsfont}
\usepackage{bm}
\usepackage{authblk}
\usepackage{todonotes}
\usepackage{enumerate}
\usepackage{mathrsfs}

\usepackage[backend=biber,style=numeric-comp,doi=true,maxnames=99]{biblatex}
\bibliography{master}

\renewcommand{\b}[1]{\boldsymbol{#1}} 
\newcommand{\bl}{{\b{\lambda}}}
\newcommand{\dl}{\mathscr{D}_{\bl}}
\newcommand*{\btau}{\bm{\tau}}
\newcommand{\sq}{\mathscr{Q}}
\newcommand{\N}{\mathbb{N}}
\newcommand{\pref}{\sqsubseteq}
\newcommand{\Z}{\mathbb{Z}}

\newcommand{\R}{\mathbb{R}}

\DeclareMathOperator{\constr}{alg}
\DeclareMathOperator{\bin}{bin}
\DeclareMathOperator{\prof}{prof}

\newtheorem{theorem}{Theorem}
\newtheorem{observation}[theorem]{Observation}
\newtheorem{lemma}[theorem]{Lemma}
\newtheorem{examp}[theorem]{Example}

\title{Algorithmic Randomness in Continuous-Time Markov Chains}

\author{Xiang Huang}
\affil{University of Illinois Springfield}

\author{Jack H. Lutz}
\affil{Iowa State University}

\author{Neil Lutz}
\affil{University of Notre Dame}

\author{Andrei N. Migunov}
\affil{Drake University}

\date{}

\begin{document}
\maketitle

\begin{abstract}
In this paper we develop the elements of the theory of algorithmic randomness in continuous-time Markov chains (CTMCs).  Our main contribution is a rigorous, useful notion of what it means for an \emph{individual trajectory} of a CTMC to be {\it random}.  CTMCs have discrete state spaces and operate in continuous time.  This, together with the fact that trajectories may or may not halt, presents challenges not encountered in more conventional developments of algorithmic randomness.

Although we formulate algorithmic randomness in the general context of CTMCs, we are primarily interested in the \emph{computational} power of stochastic chemical reaction networks, which are special cases of CTMCs.  This leads us to embrace situations in which the long-term behavior of a network depends essentially on its initial state and hence to eschew assumptions that are frequently made in Markov chain theory to avoid such dependencies.

After defining the randomness of trajectories in terms of martingales (algorithmic betting strategies), we prove equivalent characterizations in terms of algorithmic measure theory and Kolmogorov complexity.  As a preliminary application we prove that, in any stochastic chemical reaction network, \emph{every} random trajectory with bounded molecular counts has the \emph{non-Zeno property} that infinitely many reactions do not occur in any finite interval of time.
\end{abstract}

\section{Introduction}\label{sec:intro}
Stochastic chemical reaction networks are used in molecular programming, DNA nanotechnology, and synthetic biology to model and specify the behaviors of natural and engineered molecular systems.  Stochastic chemical reaction networks are known to be Turing universal \cite{jSoCoWiBr08}, hence capable of extremely complex dynamical behavior.

Briefly and roughly (deferring details until later in the paper), a stochastic chemical reaction network $N$ is a mathematical model of a chemical process in a volume $V$ of solution.  A state of $N$ consists of the nonnegative integer populations of each of its finitely many species (types of molecules) at a given time. The state space is thus countable and discrete.  The network stays in a state for a positive, real-valued sojourn time after which one of the finitely many reactions that $N$ allows to occur among its species produces an instantaneous jump transition to a different state.  Both the sojourn time and the choice of the reaction are probabilistic, with the network behaving as a certain kind of continuous-time Markov chain given by the parameters of $N$.  Hence, given an initial state at time $t=0$, there are in general uncountably many trajectories (sequences of states and sojourn times) that $N$ can traverse.  Some of these trajectories are finite (because $N$ reaches a state in which none of its reactions can occur), and some are infinite.

In this paper we develop the elements of the theory of algorithmic randomness in continuous-time Markov chains (CTMCs).  Specifically, our main contribution is a rigorous, useful notion of what it means for an {\it individual trajectory} (also called a single orbit in dynamical systems theory) of a CTMC $C$ to be {\it random} with respect to $C$ and an initial state---or probability distribution of initial states---of $C$.  This is a first step toward carrying out {\it Kolmogorov's program} of replacing probabilistic laws stating that {\it almost every} trajectory has a given property with stronger randomness laws stating that {\it every random} trajectory has the property.  More generally, we are initiating an algorithmic ``single orbit'' approach (in the sense of Weiss \cite{Weis00}) to the dynamics of CTMCs.  In a variety of contexts ranging from Bernoulli processes to ergodic theory, Brownian motion, and algorithmic learning, this algorithmic single-orbit approach has led to improved understanding of known results \cite{oLiVit19,downey2010algorithmic,nies2009computability,shen2017kolmogorov,v1998ergodic,nandakumar2008effective,fouche2009fractals,jKjoNer09,allen2015zeros,fouche2015kolmogorov,vovk2022algorithmic,shafer2008tutorial,ghosh2012predictive,vovk1999machine}.  In the context of fractal geometry, this approach has even led to recent solutions of classical open problems whose statements did not involve algorithms or single orbits \cite{CCLLMS26,BusFie25,FieStu23,FieStu24a,FieStuUniversal,Fiedler2026-packing,jLutStu20,jLutz21,LutStu24,Sla21,StullPinned,AlBuWi2025,StullOptimal}.
  
The fact that CTMCs have discrete state spaces and operate in continuous time, together with the fact that trajectories may or may not halt, presents challenges not encountered in more conventional developments of algorithmic randomness.  Our formulation of randomness is nevertheless general.  Because we are interested in the {\it computational} power of stochastic chemical reaction networks, we embrace situations in which the long-term behavior of a network depends essentially on its initial state.  Our development thus does not make assumptions that are frequently used in Markov chain theory to avoid such dependencies.

Our approach is also general in another sense, one involving Kolmogorov's program. Once one has succeeded in replacing an ``almost every'' probabilistic law with an ``every random''  law, a natural next question is, ``{\it How much} randomness is sufficient for the latter?''  Saying that an individual object is random is saying that it ``appears random'' to a class of computations.  Roughly speaking, an object is algorithmically random (or Martin-L\"{o}f random) if it appears random to all computably enumerable sets. But weaker notions of randomness such as computable randomness, polynomial-space randomness, polynomial-time randomness, and finite-state randomness, have also been extensively investigated.  Three examples of answers to the ``how much randomness suffices'' question in the context of infinite binary sequences are that (i) every algorithmically random sequence satisfies Birkhoff's ergodic theorem \cite{v1998ergodic}; (ii) every polynomial-time random sequence satisfies the Khinchin-Kolmogorov law of the iterated logarithm \cite{wang1996randomness}; and (iii) every finite-state random sequence satisfies the strong law of large numbers \cite{SchSti72}.

Although we are primarily concerned with algorithmic randomness in the present paper, we want our randomness notion to be general enough to extend easily to other computational ``levels'' of randomness, so that ``how much randomness'' questions can be formulated and hopefully answered. For this reason, we define algorithmic randomness in CTMCs using the
martingale (betting strategy) approach of Schnorr \cite{DBLP:journals/mst/Schnorr71}.  This approach extends to other levels of randomness in a straightforward manner, while our present state (i.e., lack) of knowledge in computational complexity theory does not allow us to extend other approaches (e.g., Martin-L\"{o}f tests or Kolmogorov complexity, which are known to be equivalent to the martingale approach at the algorithmic level \cite{oLiVit19,downey2010algorithmic,nies2009computability,shen2017kolmogorov}) to time-bounded complexity classes.

We develop our algorithmic randomness theory in stages.  In section~\ref{sec:bts} we develop the underlying qualitative structure of {\it Boolean transition systems}, defined so that (i) state transitions are nontrivial, i.e., not from a state to itself, and (ii) trajectories may or may not terminate.  We then show how to use these transition systems to model rate-free chemical reaction networks.

In section~\ref{sec:rss} we add probabilities, thereby defining {\it probabilistic transition systems}.  For each probabilistic transition system $\mathscr{Q}$ and each initialization $\sigma$ of $\mathscr{Q}$ we then define $(\mathscr{Q}, \sigma)$-{\it martingales}, which are strategies for betting on the successive entries in a sequence of states of $(\mathscr{Q}, \sigma)$.  Following the approach of Schnorr \cite{DBLP:journals/mst/Schnorr71}, we then define a maximal state sequence $\b{q}$ of $(\mathscr{Q}, \sigma)$ to be {\textit{random}} if there is no lower semicomputable $(\mathscr{Q}, \sigma)$-martingale that {\it succeeds} on $\b{q}$, i.e., makes unbounded money betting along $\b{q}$.  This notion of randomness closely resembles the well-understood theory of random sequences over a finite alphabet \cite{oLiVit19,downey2010algorithmic,nies2009computability,shen2017kolmogorov}, except that here the state set may be countably infinite; transitions from a state to itself are forbidden; and a positive-probability state sequence may terminate, in which case it is random.

Section~\ref{sec:rsst} is where we confront the main challenge of algorithmic randomness in CTMCs, the fact that they operate in continuous, rather than discrete, time. There we develop the algorithmic randomness of sequences $\b{t} = (t_0, t_1, \ldots)$ of sojourn times $t_i$ relative to corresponding sequences ${\bl} = (\lambda_0, \lambda_1, \ldots)$ of nonnegative real-valued rates $\lambda_i$. Each $\lambda_i$ in such a sequence is regarded as defining an exponential probability distribution function $F_{\lambda_i}$, and the sojourn times $t_i$ are to be independently random relative to these.  We use a careful binary encoding of sojourn times to define ${\bl}$-{\it martingales} that bet along sequences of sojourn times, and we again follow the Schnorr approach, defining a sequence $\b{t}$ of sojourn times to be ${\bl}$-{\it random} if there is no lower semicomputable ${\bl}$-martingale that succeeds on it.

In section~\ref{sec:mmctmc} we put the developments of sections 3 and 4 together. A {\it trajectory} of a continuous-time Markov chain $C$ is a sequence {$\btau$} of ordered pairs $(q_n,t_n)$, where $q_n$ is a state of $C$ and $t_n$ is the sojourn time that $C$ spends in state $q_n$ before jumping to state $q_{n+1}$.  For each continuous-time Markov chain $C$, we define the notion of a $C$-{\it martingale}. In section~\ref{sec:rctmct}, following Schnorr once again, we define a trajectory $\btau$ of $C$ to be {\it random} if no lower semicomputable martingale succeeds on it. We prove analogs of Schnorr's theorem and van Lambalgen's theorem for CTMC trajectories. In section 7, we give a Kolmogorov complexity characterization of the randomness of trajectories of continuous-time Markov chains. As an example application, we prove in section 8 that, in any stochastic chemical reaction network, {\it every} random trajectory $\btau$ with bounded molecular counts has the {\it non-Zeno property} that infinitely many reactions do not occur in any finite interval of time. That is, random trajectories with bounded molecular counts do not exhibit ``finite-time blowup.'' 

\section{Boolean transition systems}\label{sec:bts}
Before developing algorithmic randomness for sequences of states with respect to computable, probabilistic transition systems, we develop the underlying qualitative (not probabilistic) structure by considering transition systems that are Boolean. Some care must be taken to accommodate the fact that, in cases of interest, a sequence of states may either be infinite or end in a terminal state. 

Formally, we define a \textit{Boolean transition system} to be an ordered pair $\mathscr{Q} = (Q, \delta)$ where $Q$ is a nonempty, countable set of \textit{states}, and $\delta:Q \times Q \rightarrow \{0,1\}$ is a \textit{Boolean state transition matrix} satisfying $\delta(q,q) = 0$ for all $q \in Q$.

Intuitively, a Boolean transition system $\mathscr{Q} = (Q, \delta)$ is a nondeterministic structure that may be initialized to any nonempty set of states in $Q$. For $q,r \in Q$, the entry $\delta(q,r)$ in the Boolean transition matrix $\delta$ is the Boolean value $(0 =$ false; $1 = $ true$)$ of the condition that $r$ is reachable from $q$ in one ``step'' of $\mathscr{Q}$. The irreflexivity requirement that every $\delta(q,q) = 0$ (i.e., that $\delta$ have a zero diagonal) reflects the fact that, in all cases of interest in this paper, transitions are nontrivial changes of state. We formalize this intuition, because the formalism will be useful here.

We write $Q^*$ for the set of all finite sequences of states in $Q$, $Q^{\omega}$ for the set of all infinite sequences of states in $Q$, and $Q^{\leq \omega} = Q^* \cup Q^{\omega}$. The \textit{length} of a sequence $\b{q} \in Q^{\leq \omega}$ is

\[
|\b{q}| = \begin{cases}
	\ell &\text{ if } \b{q} = (q_0, q_1, \ldots, q_{\ell-1}) \in Q^*\\
	\omega &\text{ if } \b{q} \in Q^\omega.
\end{cases}
\]

A sequence $\b{q} \in Q^{\leq \omega}$ can thus be written as $\b{q} = (q_i \mid i < |\b{q}|)$ in any case. We write $()$ for the \textit{empty sequence}, the sequence of length 0.

For $\b{q}, \b{r} = (r_i \mid i < |\b{r}|) \in Q^{\leq \omega}$, we say that $\b{q}$ is a \textit{prefix} of $\b{r}$, and we write $\b{q} \sqsubseteq \b{r}$, if $|\b{q}| \leq |\b{r}|$ and $\b{q} = (r_i \mid i < |\b{q}|)$. It is easy to see that $\sqsubseteq$ is a partial ordering of $Q^{\leq \omega}$. 

An \textit{initialization} of a Boolean transition system $\mathscr{Q} = (Q, \delta)$ is a Boolean-valued function $\sigma: Q \rightarrow \{0,1\}$ whose \textit{support} $supp(\sigma) = \{q\in Q\mid \sigma(q) \neq 0\}$ is nonempty.

A Boolean transition system $\mathscr{Q} = (Q, \delta)$ \textit{admits} a sequence $\b{q} = (q_i \mid i < |\b{q}|) \in Q^{\leq \omega}$ with an initialization $\sigma$, and we say that $\b{q}$ is \textit{$\mathscr{Q}$-admissible} from $\sigma$, if the following conditions hold for all $0 \leq i < |\b{q}|$.
\begin{enumerate}
    \item If $i = 0$, then $\sigma(q_i) = 1$.
    \item If $i+1 < |\b{q}|$, then $\delta(q_i,q_{i+1}) = 1$.
\end{enumerate}

A sequence $\b{q} \in Q^{\leq \omega}$ that is $\mathscr{Q}$-admissible from $\sigma$ is \textit{maximal} if, for every sequence $\b{r} \in Q^{\leq \omega}$ that is $\mathscr{Q}$-admissible from $\sigma$, $\b{q} \sqsubseteq \b{r} \implies \b{q} = \b{r}$.

We use the following notations.
\begin{align*}
    Adm_{\mathscr{Q}}(\sigma) &= \{\b{x} \in Q^* \mid \b{x}\text{ is }\mathscr{Q}\text{-admissible from }\sigma\}.\\
    \mathds{A}_{\mathscr{Q}}(\sigma) &= \{\b{q} \in Q^{\leq\omega} \mid \b{q}\text{ is a maximal }\mathscr{Q}\text{-admissible sequence from }\sigma\}.
\end{align*}

When $\mathscr{Q}$ is obvious from the context, we omit it from the notation and write these sets as $Adm(\sigma)$ and $\mathds{A}(\sigma)$. Note that elements of $Adm_\mathscr{Q}(\sigma)$ are required to be \textit{finite} sequences.

Intuitively, $\mathds{A}_{\mathscr{Q}}(\sigma)$ is the set of all possible ``behaviors'' of the Boolean transition system $\mathscr{Q} = (Q,\delta)$ with the state initialization $\sigma:Q \rightarrow \{0,1\}$. The fact that $\delta$ is irreflexive implies that $q_i \neq q_{i+1}$ holds for all $i \in \mathbb{N}$ such that $i+1 < |\b{q}|$ in every admissible sequence $\b{q} = (q_i \mid i < |\b{q}|) \in \mathds{A}_{\mathscr{Q}}(\sigma)$. In this paper we do not regard the indices $i = 0, 1, \ldots$ in a state sequence $\b{q} = (q_0, q_1, \ldots)$ as successive instants in discrete time. In our main applications, the amount of time spent in state $q_i$ varies randomly and continuously, so it is more useful to think of the indices $i = 0, 1, \ldots$ as finite ordinal numbers, i.e., to think of $q_i$ as merely the $i^{\textit{th}}$ state in the sequence $\b{q}$.

Each $\b{x} \in Adm_\mathscr{Q}(\sigma)$ is the \textit{name} of the \textit{$\mathscr{Q}$-cylinder}
\begin{equation}\label{eq:Qcylinder}
    \mathds{A}_{\b{x}}(\sigma) = \{\b{q} \in \mathds{A}_{\mathscr{Q}}(\sigma) \mid \b{x} \sqsubseteq \b{q}\}.
\end{equation}
Each $\b{x} \in Adm(\sigma)$ is a finite---and typically partial---specification of each sequence $\b{q} \in \mathds{A}_{\b{x}}(\sigma)$. The collection
\[\mathscr{A}(\sigma) = \mathscr{A}_{\mathscr{Q}}(\sigma) = \{ \mathds{A}_{\b{x}}(\sigma) \mid \b{x} \in Adm_\mathscr{Q}(\sigma)\}\]
is a basis for a topology on $\mathds{A}(\sigma)$. The open sets in this topology are simply the sets that are unions of (finitely or infinitely many) cylinders in $\mathscr{A}(\sigma)$. The metric (in fact, ultrametric) $d$ on $Q^{\leq\omega}$ defined by $$d(\b{q},\b{r}) = 2^{-|\b{p}|},$$ where $\b{p}$ is the longest common prefix of $\b{q}$ and $\b{r}$ (and $2^{-\infty} = 0$), induces this same topology on $\mathds{A}_{\mathscr{Q}}(\sigma)$ for each Boolean transition system $\mathscr{Q} = (Q, \delta)$ and each state initialization $\sigma: Q \rightarrow \{0,1\}$. With this topology, $\mathds{A}_{\mathscr{Q}}(\sigma)$ is a Polish space (a complete, separable metric space). This Polish space is not necessarily perfect, because it may have isolated points. Among these,
when they exist, are finite sequences, i.e., sequences $x \in Q^* \cap \mathds{A}_{\mathscr{Q}}(\sigma)$. Such sequences $\b{x}$ are said to \textit{halt}, or \textit{terminate}, in $\mathscr{Q}$ from $\sigma$.

A Boolean transition system $\mathscr{Q} = (Q, \delta)$ is \textit{computable} if the elements of $Q$ are naturally represented in such a way that $(i)$ the Boolean-valued function $\delta$ is computable, and $(ii)$ the set of \textit{terminal states} (i.e., states $q \in Q$ such that $\delta(q, r) = 0$ for all $r \in Q$) is decidable. An initialization $\sigma:Q \rightarrow \{0,1\}$ is \textit{computable} if its support is decidable.

An important class of examples of Boolean transition systems consists of those that model rate-free chemical reaction networks. Formally, let $\b{S} = \{X_0, X_1, X_2, \ldots\}$ be a countable set of distinct \textit{species} $X_n$, each of which we regard as an abstract type of molecule. A \textit{rate-free chemical reaction network} (or \textit{rate-free CRN}) is an ordered pair $N=(S,R)$, where $S \subseteq \b{S}$ is a finite set of species, and $R$ is a finite set of (\textit{rate-free}) \textit{reactions} on $S$, each of which is formally an ordered pair $\rho = (r, p)$ of distinct vectors $r, p \in \mathbb{N}^S$ (equivalently, functions $r,p:S \rightarrow \mathbb{N}$). Informally, we write species in notations convenient for specific problems ($X, Y, Z, \widehat{X}, \overline{Y}$, etc.) rather than as subscripted elements of $\b{S}$, and we write reactions in a notation more suggestive of chemical reactions. For example, 
\begin{equation}\label{raction:example}
X + Z \rightarrow 2Y + Z
\end{equation}
is a rate-free reaction on the set $S = \{X, Y, Z\}$. If we consider the elements of $S$ to be ordered as written, then the left-hand side of (\ref{raction:example}) is formally the \textit{reactant vector} $r = (1, 0, 1)$, and the right-hand side of (\ref{raction:example}) is the \textit{product vector} $p = (0, 2, 1)$. A species $Y \in S$ is called a \textit{reactant} of a reaction $\rho = (r, p)$ if $r(Y) > 0$ and a \textit{product} of $\rho$ if $p(Y) > 0$.

Intuitively, the reaction $\rho$ in (\ref{raction:example}) means that, if a molecule of species $X$ encounters a molecule of species $Z$, then the reaction $\rho$ may occur, in which case the reactants $X$ and $Z$ disappear and the products---two molecules of species $Y$ and a molecule of species $Z$---appear in their place. Accordingly, the \textit{net effect} of a reaction $\rho = (r,p)$ is the vector $\Delta\rho \in \mathbb{Z}^S$ defined by
\begin{equation*}
\Delta\rho(Y) = p(Y) - r(Y)
\end{equation*}
for all $Y \in S$. Since we have required $r$ and $p$ to be distinct, $\Delta\rho$ is never the zero vector in $\mathbb{Z}^S$. 

In this paper, a \textit{state} of a chemical reaction network $N = (S,R)$ is a vector $q \in \mathbb{N}^S$. Intuitively, $N$ is modeling chemical processes in a solution, and the state $q$ denotes a situation in which, for each $Y \in S$, exactly $q(Y)$ molecules of species $Y$ are present in the solution.

A reaction $\rho = (r,p) \in R$ of a chemical reaction network $N = (S,R)$ \textit{can occur} (or is \emph{enabled}) in a state $q \in \mathbb{N}^S$ if 
\begin{equation*}
q(Y) \geq r(Y)
\end{equation*}
 holds for every $Y \in S$, i.e. if the reactants of $\rho$ are present in $q$. If this reaction $\rho$ \textit{does} occur in state $q$, then it transforms $q$ to the new state $q + \Delta\rho$. 

The behavior of a rate-free chemical reaction network $N = (S,R)$ clearly coincides with that of the Boolean transition system $\mathscr{Q}_N = (\mathbb{N}^S, \delta)$, where $\delta: \mathbb{N}^S \times \mathbb{N}^S \rightarrow \{0,1\}$ is defined by setting each $\delta(q, q')$ to be the Boolean value of the condition that some reaction $\rho \in R$ transforms the state $q$ to the state $q'$. Boolean transition systems of this form are clearly computable and have other special properties. As one example, for each $q \in \mathbb{N}^S$, there are only finitely many $q' \in \mathbb{N}^S$ for which $\delta(q, q') = 1$. 

Rate-free chemical reaction networks, and Boolean transition systems more generally, raise significant and deep problems in distributed computing \cite{leroux2015demystifying,czerwinski2019reachability}, but our focus here is on randomness, which we introduce in the following section.

\section{Random state sequences}\label{sec:rss}
This section develops the elements of algorithmic randomness for sequences of states with respect to computable, probabilistic transition rules. The probabilistic transition systems defined in this section, in which self-transitions have probability zero, correspond to the state sequences of continuous-time Markov chains.

Formally, we define a \textit{probabilistic transition system} to be an ordered pair $\mathscr{Q} = (Q,\pi)$, where $Q$ is a countable set of \textit{states}, and $\pi:Q \times Q \rightarrow [0,1]$ is a \textit{probabilistic transition matrix}, by which we mean that $\pi$ satisfies the following two conditions for each state $q \in Q$. 
\begin{enumerate}
\item $\pi(q,q) = 0$.
\item The sum $\pi(q) = \sum_{r\in Q} \pi(q, r)$ is either 0 or 1.
\end{enumerate}

If $\pi(q)=0$, then $q$ is a \textit{terminal state}. If $\pi(q)=1$, then $q$ is a \textit{nonterminal state}. 

If $\mathscr{Q} = (Q, \pi)$ is a probabilistic transition system, and we define $\delta: Q \times Q \rightarrow \{0,1\}$ by $$\delta(q, r) = sgn(\pi(q,r))$$ for all $q, r \in Q$, where $sgn:[0,\infty) \rightarrow \{0,1\}$ is the signum function 
\[
    sgn(x)= \begin{cases}
    	0 &\textit{if } x = 0\\
    	1 &\textit{if } x > 0,
    \end{cases}
\]
then $\mathscr{Q}_B = (Q, \delta)$ is the Boolean transition system \textit{corresponding} to $Q$. The essential difference between $\mathscr{Q}_B$ and $\sq$ is that, while $\delta(q,r)$ merely says whether it is \textit{possible} for $\sq_B$ (or $\sq$) to transition from $q$ to $r$ in one step, $\pi(q,r)$ is the \textit{quantitative probability} of doing so. 

An \textit{initialization} of a probabilistic transition system $\sq = (Q, \pi)$ is a discrete probability measure $\sigma$ on $Q$, i.e., a function $\sigma: Q \rightarrow [0,1]$ satisfying $\sum_{q \in Q} \sigma(q) = 1$. The \textit{Boolean version} of such an initialization $\sigma$ is the function $\sigma_B:Q \rightarrow \{0,1\}$ defined by $$\sigma_B(q) = sgn(\sigma(q))$$ for each $q \in Q$. It is clear that $\sigma_B$ is an initialization of $\sq_B$.

Given a probabilistic transition system $\sq = (Q, \pi)$ and an initialization $\sigma$ of $\sq$, we define the sets 

\begin{align*}
    Adm(\sigma) &= Adm_\sq(\sigma) = Adm_{\sq_B}(\sigma_B)\\
    \mathds{A}(\sigma) &= \mathds{A}_{\mathscr{Q}}(\sigma) = \mathds{A}_{\sq_B}(\sigma_B),
\end{align*}
relying on the fact that the right-hand sets were defined in section~\ref{sec:bts}. The notations and terminology in section~\ref{sec:bts} leading up to these definitions are similarly extended to probabilistic transition systems, as are the definitions of the $\sq$-cylinders $\mathds{A}_{\b{x}}(\sigma)$ and the basis $\mathscr{A}(\sigma)$ for the topology $\mathds{A}(\sigma)$. 

What we can do here that we could not do for Boolean transition systems is define a Borel probability measure on each set $\mathds{A}_{\mathscr{Q}}(\sigma)$. Specifically, for each probabilistic transition system $\sq = (Q, \pi)$ and each initialization $\sigma$ of $\sq$, define the function
\[\mu_{\sq, \sigma}: Adm_\sq(\sigma) \rightarrow [0,1]\]
as follows. Let $\b{x} = (x_i \mid i < |\b{x}|) \in Adm_\sq(\sigma)$. If $|\b{x}| = 0$, then $\mu_{\sq,\sigma}(x) = 1$. If $|\b{x}| > 0$, then 
\begin{equation*}\mu_{\sq,\sigma}(\b{x}) = \sigma(x_0){\displaystyle \prod_{i=0}^{|\b{x}|-2}\pi(x_i, x_{i+1})}. \end{equation*}

Since $\b{x}$ is a name of the cylinder $\mathds{A}_{\mathscr{Q},\b{x}}(\sigma)$, each $\mu_{\sq,\sigma}(\b{x})$ here should be understood as an abbreviation of $\mu_{\sq,\sigma}(\mathds{A}_{\b{x}}(\sigma))$, which is intuitively the probability that an element of $\mathds{A}_{\sq,\b{x}}(\sigma)$ begins with the finite sequence $\b{x}$. 

\begin{observation}
	If a sequence $\b{x} \in Adm_{\sq}(\sigma)$ does not terminate, then \begin{equation*}\mu_{\sq,\sigma}(\b{x}) = \sum_{\substack{\b{x} \sqsubseteq \b{y} \in Adm_\sq(\sigma)\\|\b{y}| = |\b{x}| + 1}} \mu_{\sq,\sigma}(\b{y}). \end{equation*}
\end{observation}

The above observation implies that $\mu_{\sq,\sigma}$ can, by standard techniques (e.g., Dynkin's $\pi$-$\lambda$ theorem \cite{billingsley1995probability}), be extended to a Borel probability measure on $\mathds{A}_{\sq}(\sigma)$, i.e., to a function $\mu_{\sq,\sigma}$ that assigns probability $\mu_{\sq,\sigma}(E)$ to every Borel set $E \subseteq \mathds{A}_{\sq}(\sigma)$.

An arbitrary set $E\subseteq\mathds{A}_{\sq}(\sigma)$ has \emph{probability 0}, and we write $\mu_{\sq,\sigma}(E)=0$, if there is some Borel set $F\supseteq E$ such that $\mu_{\sq,\sigma}(F)=0$. It has \emph{algorithmic probability 0}, and we write $\mu_{\sq,\sigma,\constr}(E)=0$, if there is a computable function $g:\N\times\N\to Adm_\sq(\sigma)$ such that, for all $k\in\N$,
\[E\subseteq\bigcup_{\ell\in\N}\mathds{A}_{g(k,\ell)}(\sigma)\]
and
\[\sum_{\ell\in\N}\mu_{\sq,\sigma}(g(k,\ell))\leq 2^{-k}.\]

    If $\sq$ is a probabilistic transition system and $\sigma$ is an initialization of $\sq$, then a $(\sq, \sigma)$-\emph{martingale} is a function $$d:Adm_\sq(\sigma) \rightarrow [0, \infty)$$ such that, for every non-terminating sequence $\b{x} \in Adm_\sq(\sigma)$,
	\begin{equation}\label{eq:state_martingale}
		d(\b{x})\mu_{\sq,\sigma}(\b{x}) = \sum_{\substack{\b{x} \sqsubseteq \b{y} \in Adm_\sq(\sigma)\\|\b{y}| = |\b{x}| + 1}} d(\b{y})\mu_{\sq, \sigma}(\b{y}).
	\end{equation}

Intuitively, a $(\sq, \sigma)$-martingale $d$ is a gambler that bets on the successive states in a sequence
\[\b{q} = (q_i \mid i < |\b{q}|) \in \mathds{A}_{\sq}(\sigma).\]
The gambler's initial capital is $d(())$, and its capital after betting on a prefix $\b{x} \in Adm_\sq(\sigma)$ of $\b{q}$ is $d(\b{x})$. The condition (\ref{eq:state_martingale}) says that the payoffs are fair with respect to the probability measure $\mu = \mu_{\sq,\sigma}$ in the sense that the conditional expectation of the gambler's capital after betting on the state following $\b{x}$ in $\b{q}$, given that $\b{x} \sqsubseteq \b{q}$, is exactly the gambler's capital before placing this bet.

A $(\sq, \sigma)$-martingale $d$ \emph{succeeds} on a sequence $\b{q} \in \mathds{A}_{\sq}(\sigma)$ if the set $$\{d(\b{x}) \mid \b{x} \in Adm_\sq(\sigma) \text{ and } \b{x} \sqsubseteq \b{q}\}$$ is unbounded. The \emph{success set} of a $(\sq, \sigma)$-martingale $d$ is $$S^\infty[d] = \{ \b{q} \in \mathds{A}_{\sq}(\sigma) \mid d \text{ succeeds on } \b{q} \}.$$

Following standard practice, we develop randomness by imposing computability conditions on martingales. Recall that, if $D$ is a discrete domain, then a function $f:D \rightarrow \mathbb{R}$ is \textit{computable} if there is a computable function $\hat{f}:D \times \mathbb{N} \rightarrow \mathbb{Q}$ such that, for all $x \in D$ and $r \in \mathbb{N}$, $$|\hat{f}(x,r) - f(x)| \leq 2^{-r}.$$ The parameter $r$ here is called a \textit{precision parameter}.

A function $f: D \rightarrow \mathbb{R}$ is \textit{lower semicomputable} if there is a computable function $\hat{f}:D \times \mathbb{N} \rightarrow \mathbb{Q}$ such that the following two conditions hold for all $x \in D$.
\begin{enumerate}
    \item For all $s \in \mathbb{N}, \hat{f}(x,s) \leq \hat{f}(x, s+1) < f(x)$.
    \item $\displaystyle\lim_{s \rightarrow \infty} \hat{f}(x,s) = f(x)$.
\end{enumerate}
The parameter $s$ is sometimes called a \textit{patience parameter}, because the convergence in the second condition can be very slow. 

A probabilistic transition system $\sq = (Q, \pi)$ is \textit{computable} if the elements of $Q$ are naturally represented in such a way that (i) the probability transition matrix $\pi:Q \times Q \rightarrow [0,1]$ is computable in the above sense, and (ii) the support of $\pi$ and the set of terminal states are decidable. (It is well known that (ii) does not follow from (i)~\cite{oKo91,oWeih00}. Fortunately, (ii) does hold in many cases of interest, including chemical reaction networks.)

Similarly, an initialization $\sigma$ of a probabilistic transition system $\sq = (Q, \pi)$ is \textit{computable} if (i) the function $\sigma:Q \rightarrow [0,1]$ is computable, and (ii) the support of $\sigma$ is decidable.

Let $\sq$ be a probabilistic transition system that is computable, and let $\sigma$ be an initialization of $\sq$ that is also computable. A state sequence $\b{q} \in \mathds{A}_{\sq}(\sigma)$ is (\textit{algorithmically}) \textit{random} if there is no lower semicomputable $(\sq, \sigma)$-martingale that succeeds on $\b{q}$.

This notion of random sequences in $\mathds{A}_{\sq}(\sigma)$ closely resembles the well-understood theory of random sequences on a finite alphabet \cite{zvonkin1970complexity, schnorr1977survey}. The main differences are that here the state set may be countably infinite; transitions from a state to itself are forbidden; and a positive-probability state sequence may terminate, in which case it is clearly random. The following analogs of the theorems of Ville and Schnorr hold for probabilistic transition systems. 

\begin{theorem}[\normalfont Ville \cite{Ville39}]\label{ville_state}
	Let $\sq$ be a probabilistic transition system, and let $\sigma$ be an initialization of $\sq$. For every set $E \subseteq \mathds{A}_{\sq}(\sigma)$, the following two conditions are equivalent.
	\begin{enumerate}
        \item $\mu_{\sq, \sigma}(E) = 0$.
        \item There is a $(\sq, \sigma)$-martingale $d$ such that $E \subseteq S^\infty[d]$.
    \end{enumerate}
\end{theorem}

\begin{theorem}[\normalfont Schnorr \cite{DBLP:journals/mst/Schnorr71}]\label{schnorr_state}
    Let $\sq$ be a computable probabilistic transition system and let $\sigma$ be a computable initialization of $\sq$. For every set $E \subseteq \mathds{A}_{\sq}(\sigma)$, the following two conditions are equivalent.
    \begin{enumerate}
        \item $\mu_{\sq, \sigma,\constr}(E) = 0$.
        \item There is a lower semicomputable $(\sq, \sigma)$-martingale $d$ such that $E \subseteq S^\infty[d]$.
    \end{enumerate}
\end{theorem}

\section{Random sequences of sojourn times}\label{sec:rsst}

The ``sojourn time'' that a continuous-time Markov chain spends in a state before jumping to a new state may be any element of $(0,\infty]$, i.e., any duration $t$ that is either a (strictly) positive real number or $\infty$. This section thus develops the elements of algorithmic randomness for sequences of durations $t \in (0,\infty]$ with respect to sequences of probability measures that occur in continuous-time Markov chains. 

A \textit{rate} in this paper is a nonnegative real number $\lambda \in [0,\infty)$. We rely on context to distinguish this standard use of $\lambda$ from the equally standard use of $\lambda$ to denote the empty string.

We interpret each rate $\lambda >0 $ as a name of the exponential probability measure with rate $\lambda$, i.e., the probability measure on $(0,\infty]$ whose cumulative distribution function $F_\lambda:(0,\infty] \rightarrow [0,1]$ is given by $$F_\lambda(t) = 1-e^{-\lambda t}$$ for all $t \in (0,\infty]$, where $e^{-\infty} = 0$. We interpret the rate $\lambda = 0$ as a name of the point-mass probability on $(0, \infty]$ that concentrates all the probability at $\infty$. This has the cumulative distribution function $F_0:(0,\infty] \rightarrow [0,1]$ given by 

\[
F_0(t) = \begin{cases}
	0 &\text{ if } t \in (0,\infty)\\
	1 &\text{ if } t = \infty.
\end{cases}
\]

We associate each string $w \in \{0,1\}^*$ with the interval $I_w \subseteq [0,1]$ defined as follows. Let $w$ be the lexicographically $i^{\text{th}}$ $(0 \leq i < 2^{|w|})$ element of $\{0,1\}^{|w|}$ where $0^{|w|}$ is the $0^{\text{th}}$ element and $1^{|w|}$ is the $(2^{|w|} -1)^{\text{st}}$ element. Then 
$$I_w = (2^{-|w|}i, 2^{-|w|}(i+1)].$$
Note that, for each $w \in \{0,1\}^*$ and $\ell \in \mathbb{N}$, the intervals $I_{wu}$, for $u \in \{0,1\}^\ell$, form a \textit{left-to-right partition} of $I_w$, i.e., a partition of $I_w$ in which $I_{wu}$ lies to the left of $I_{wv}$ if and only if $u$ lexicographically precedes $v$.

For each rate $\lambda \in [0,\infty)$ and each string $w \in \{0,1\}^*$, define the interval $$D_\lambda(w) = F_\lambda^{-1}(I_w) \subseteq (0,\infty].$$
\begin{examp}\label{example:dist}
	If $\lambda > 0$, then 
	\begin{align*}
		D_\lambda(00) = (0,a_1], &\quad  D_\lambda(01) = (a_1, a_2], \\
		D_\lambda(10) = (a_2, a_3], &\quad D_\lambda(11) = (a_3, \infty],
	\end{align*}
	where $a_1 = \frac{2\ln2 - \ln3}{\lambda}$, $a_2 = \frac{\ln2}{\lambda}$, and $a_3 = \frac{2\ln2}{\lambda}$. On the other hand, $D_0(00)=D_0(01) = D_0(10) = \emptyset$, and $D_0(11) = \{\infty\}$.
\end{examp}
\begin{observation}
	If $\lambda >0$, then, for each $\ell \in \mathbb{N}$, the intervals $D_\lambda(w)$, for $w \in \{0,1\}^\ell$, form a left-to-right partition of $(0, \infty]$ into intervals that are equiprobable with respect to $F_\lambda$.

\end{observation}

Example \ref{example:dist} shows that the assumption $\lambda >0$ is essential here.

For each rate $\lambda \in [0,\infty)$, each duration $t \in (0,\infty]$, and each $w \in \{0,1\}^*$, we call $w$ a \textit{$\lambda$-approximation} (or a \textit{partial $\lambda$-specification}) of $t$, and we write $w \sqsubseteq_\lambda t$, if $t \in D_\lambda(w)$. Note that, for each duration $t\in (0,\infty]$, there is a unique sequence $\bin_{\lambda}(t)\in\{0,1\}^{\omega}$ such that, for every prefix $w\sqsubseteq \bin_{\lambda}(t)$, we have $w\sqsubseteq_{\lambda}t$. Note also that for $\lambda>0$ the function $t\mapsto \bin_{\lambda}(t)$ is a measure-preserving one-to-one correspondence between $(0, \infty]$, endowed with the probability measure $\mu_{\lambda}$ whose cumulative distribution function is $F_{\lambda}$, and the Lebesgue measure 1 subset of $\{0,1\}^{\omega}$ consisting of those sequences containing infinitely many 1s. 

For this discussion, consider a \textit{martingale} to be as originally defined by Ville \cite{Ville39}, namely, a function $d:\{0,1\}^*\to[0,\infty)$ such that, for all $w\in\{0,1\}^*$,
\[
    d(w)= \frac{d(w0)+d(w1)}{2}.
\]
Given a rate $\lambda\in [0,\infty)$, say that a martingale $d$ $\lambda$-\textit{succeeds} on a duration $t\in (0,\infty]$, and write $t\in S_{\lambda}[d]$, if the set
\[
    \{d(w)\mid w \sqsubseteq_{\lambda} t\}
\]
is unbounded. We call $S_{\lambda}[d]$ the $\lambda$-\textit{success set} of $d$. By the preceding paragraph, the following theorem follows immediately from Ville's original theorem.

\begin{theorem}\label{thm:New_thm_six}
    For every set $E\subseteq (0,\infty]$ of durations, the following two conditions are equivalent.
\end{theorem}
\begin{enumerate}
    \item $\mu_{\lambda}(E)=0.$
    \item There is a martingale $d$ such that $E\subseteq S_{\lambda}[d]$.
\end{enumerate}

Motivated by Theorem~\ref{thm:New_thm_six}, for each rate $\lambda\in [0,\infty)$, we define a duration $t \in (0,\infty]$ to be \emph{(algorithmically) $\lambda$-random} if there is no lower semicomputable martingale $d$ such that $t\in S_{\lambda}[d]$.

Our next objective is to extend the above definition to the \textit{independent} randomness of a finite sequence $(t_0, \ldots,t_{n-1})\in (0,\infty]^{n}$ of durations with respect to a corresponding sequence $(\lambda_0,\ldots, \lambda_{n-1})\in [0,\infty)^n$ of rates. There are several equivalent ways of doing this, but the simplest for our purpose here is to adapt van Lambalgen's interleaving method \cite{jLamb87} to the present context.

Given $n$ binary sequences $x_0,\ldots,x_{n-1}\in\{0,1\}^{\leq\omega}$ of equal length, define the \emph{interleaving} of these sequences $z=\bigsqcup_{r=0}^{n-1} x_r$ by
\[z[qn+r]=x_r[q]\]
for each $q<|x_0|$ and $0\leq r <n$. Intuitively, $\bigsqcup_{r=0}^{n-1} x_r$ is a ``perfect shuffle'' of the sequences $x_0,\ldots,x_{n-1}$.

For any finite rate sequence $(\lambda_0, \ldots, \lambda_{n-1})\in[0,\infty)^n$ and duration sequence $(t_0,\ldots, t_{n-1})\in (0,\infty]^n$, define the sequence
\[
	\bin_{(\lambda_0,\ldots,\lambda_{n-1})}(t_0,\ldots, t_{n-1})=\bigsqcup_{r=0}^{n-1}\bin_{\lambda_r}(t_r).
\]
Given a martingale $d$, we say the martingale $(\lambda_0, \ldots, \lambda_{n-1})$-\textit{succeeds} on the sequence $(t_0,\ldots, t_{n-1})$, and write $(t_0,\ldots, t_{n-1})\in S_{(\lambda_0,\ldots, \lambda_{n-1})}[d]$, if
\[
	\limsup_{k\to \infty}d\big(\bin_{(\lambda_0,\ldots,\lambda_{n-1})}(t_0,\ldots, t_{n-1})[0\ldots k-1]\big)=\infty.
\]
Then a finite sequence $(t_0,\ldots, t_{n-1})\in (0,\infty]^n$ of durations is \textit{independently} $(\lambda_0,\ldots,\lambda_{n-1})$-\textit{random}, where $(\lambda_0,\ldots, \lambda_{n-1})\in[0,\infty)^n$, if there is no lower semicomputable martingale $d$ such that
\[
(t_0,\ldots, t_{n-1})\in S_{(\lambda_0,\ldots,\lambda_{n-1})}[d].
\]

We now lift the above ideas to possibly infinite sequences of rates and durations.
A \textit{rate sequence} is a nonempty sequence $\b{\lambda} = (\lambda_i \mid 0 \leq i < |\b{\lambda}|) \in [0,\infty)^{\leq \omega}$ with the property that, for each $0 \leq i < |\b{\lambda}|$,
\[i+1 < |\b{\lambda}| \iff \lambda_i > 0.\]
That is, either $\b{\lambda}$ is finite with a single 0 entry, occurring at the end, or $\b{\lambda}$ is infinite with no $0$ entries.

If $\b{\lambda} = (\lambda_i \mid 0 \leq i < |\b{\lambda}|)$ is a rate sequence, then a \textit{$\bl$-duration sequence} is a sequence
\[\b{t} = (t_i \mid i < |\bl|) \in (0, \infty]^{\leq \omega}\]
such that, for each $0 \leq i < |\bl|$, $$t_i < \infty \iff \lambda_i > 0.$$ We write $\b{D_\lambda}$ for the set of all $\bl$-duration sequences. Note that 

\[
\b{D_\lambda} = \begin{cases}
		(0,\infty)^{|\bl| -1} \times \{\infty\} &\text{if } |\bl| < \omega.\\
	(0,\infty)^\omega &\text{if } |\bl| = \omega
\end{cases}
\]
depends only on the length of $\bl$, not on the components of $\bl$.

If $\b{\lambda} = (\lambda_i \mid 0 \leq i < |\b{\lambda}|)$ is a rate sequence, $\b{t} = (t_i \mid i < |\bl|) \in \b{D_\lambda}$ is a $\bl$-duration sequence, and $\b{w} = (w_i \mid i < |\b{w}|) \in (\{0,1\}^*)^*$ is a finite sequence of binary strings with $|\b{w}| \leq |\bl|$, then we call $\b{w}$ a \textit{$\bl$-approximation} (or a \textit{partial $\bl$-specification}) of $\b{t}$, and we write $\b{w} \sqsubseteq_{\bl} \b{t}$, if $w_i \sqsubseteq_{\lambda_i} t_i$ holds for all $0 \leq i < |\b{w}|$.

If $\bl$ is a rate sequence and $\b{w} \in (\{0,1\}^*)^*$ is a finite sequence of binary strings with $|\b{w}| \leq |\bl|$, then the \textit{$\bl$-cylinder generated by $\b{w}$} is the set
\[\b{D_\lambda}(\b{w}) = \{\b{t} \in \b{D_\lambda} \mid \b{w} \sqsubseteq_{\bl} \b{t}\}.\]

Given a rate sequence $\bl$, let
\[\Phi=\Phi[\bl]=\bigcup_{n\in\N}(\{0,1\}^n)^{\min\{n,|\bl|\}}.\]
It is routine to verify that the collection
\[\dl = \{\b{D_\lambda}(\b{w}) \mid \b{w} \in \Phi\}\]
is a semi-algebra of subsets of $\b{D_\lambda}$ that generates the $\sigma$-algebra $\mathscr{B}_\bl$ of all Borel subsets of $\b{D}_{\bl}$. Defining the function $\mu_\bl:\dl \rightarrow [0,1]$ by
\[\mu_\bl(\b{D}_{\bl}(\b{w})) =
    \begin{cases}
        2^{-|\b{w}|^2}&\text{if }|\b{w}|<|\bl|\\
        2^{-(|\bl|-1)|w_0|}&\text{if }|\bm{w}|=|\bl|\text{ and }w_{|\bl|-1}=1^{|w_0|}\\
        0&\text{otherwise}.
    \end{cases}
\]
for all $\b{w}=(w_0,\ldots,w_{n-1}) \in \Phi$, it follows by standard techniques that $\mu_\bl$ extends uniquely to a probability measure $\mu_\bl: \mathscr{B}_\bl \rightarrow [0,1]$.
When convenient, we use the abbreviation
\[\mu_\bl(\b{w}) = \mu_\bl(\b{D_\bl}(\b{w})).\]

An arbitrary set $E\subseteq\b{D}_\bl$ has \emph{probability 0}, and we write $\mu_\bl(E)=0$, if there is some Borel set $F\supseteq E$ such that $\mu_\bl(F)=0$. It has \emph{algorithmic probability 0}, and we write $\mu_{\bl,\constr}(E)=0$, if there is a computable function $g:\N\times\N\to\Phi$ such that, for all $k\in\N$,
\[E\subseteq\bigcup_{\ell\in\N}\b{D}_\bl(g(k,\ell))\]
and 
\[\sum_{\ell\in\N}\mu_\bl(g(k,\ell))\leq 2^{-k}.\]

If $\b{\lambda}$ is a rate sequence, then a \textit{$\bl$-martingale} is a function
\[d:\Phi \rightarrow [0,\infty)\]
that satisfies
\[d(\b{w})\mu_\bl(\b{w}) = \sum_{\substack{b_0,\ldots,b_n\in\{0,1\}\\w_n\in\{0,1\}^n}}d(w_0b_0,\ldots,w_{n}b_{n})\mu_\bl(w_0b_0,\ldots,w_{n}b_{n})\]
for all $\b{w} = (w_0, \ldots, w_{n-1}) \in\Phi$ such that $n<|\bl|$ and
\[d(\b{w})\mu_\bl(\b{w}) = \sum_{b_0,\ldots,b_{n-1}\in\{0,1\}}d(w_0b_0,\ldots,w_{n-1}b_{n-1})\mu_\bl(w_0b_0,\ldots,w_{n-1}b_{n-1})\]
for all $\b{w} = (w_0, \ldots, w_{n-1}) \in\Phi$ such that $n=|\bl|$.

Intuitively, a $\bl$-martingale $d$ is a strategy that a gambler may use for betting on approximations $w_i$ of the durations $t_i$ in a $\bl$-duration sequence $\b{t} = (t_i \mid i < |\b{t}|)$. The gambler's initial amount of money is the value $d(())$ of $d$ at the empty sequence $()$ of binary strings. If $\b{w} = (w_0, \ldots, w_{n-1}) \sqsubseteq_\bl \b{t}$, meaning that each $t_i$ is in the interval $\b{D_\bl}(w_i) \subseteq (0,\infty]$, then $d(\b{w})$ is the amount of money that the gambler has after betting on $\b{w}$.

A $\bl$-martingale $d$ \textit{succeeds} on a $\bl$-duration sequence $\b{t}$ if the set
\[\{d(\b{w}) \mid \b{w}\in\Phi\text{ and }\b{w} \sqsubseteq_\bl \b{t} \}\]
is unbounded. The \textit{success set} of a $\bl$-martingale $d$ is
\[S^\infty[d] = \{\b{t} \in \b{D_\bl} \mid d \text{ succeeds on } \b{t}\}.\]                         
We define a sequence $\b{t} \in \b{D_{\lambda}}$ to be \emph{(algorithmically) $\bl$-random} if there is no lower semicomputable $\bl$-martingale that succeeds on $\b{t}$.

\begin{theorem}\label{thm:indrand}
  If $\b\lambda$ is a rate sequence and $\b{t} = (t_0,t_1,\ldots) \in \b{D_{\lambda}}$ is $\b\lambda$-random, then for all $n\leq|\b\lambda|$, the finite sequence $(t_0,\ldots,t_{n-1})$ is independently $(\lambda_0,\ldots,\lambda_{n-1})$-random. The converse does not hold.
\end{theorem}

\begin{proof}
We prove the contrapositive. Let $\b\lambda$ be a rate sequence, let $\b{t} = (t_0,t_1,\ldots) \in \b{D_{\lambda}}$, and suppose there is some $n\leq |\b\lambda|$ such that $(t_0,\ldots,t_{n-1})$ is not independently $(\lambda_0,\ldots,\lambda_{n-1})$-random. Let
\[z=\bin_{(\lambda_0,\ldots,\lambda_{n-1})}(t_0,\ldots,t_{n-1}).\]
Then there is a lower semicomputable martingale $d$ such that $(t_0,\ldots,t_{n-1})\in S_{(\lambda_0,\ldots,\lambda_{n-1})}[d]$, meaning the set
\[\left\{d(w)\mid w\sqsubseteq z\right\}\]
is unbounded. Define the function $d':\Phi\to[0,\infty)$ by
\[d'(w_0b_0,\ldots,w_{\ell-1}b_{\ell-1})=
    \begin{cases}
        1&\text{if }\ell<n\\
        \frac{d\left(\bigsqcup_{r=0}^{n-1} w_rb_r\right)}{d\left(\bigsqcup_{r=0}^{n-1} w_r\right)}&\text{if }\ell\geq n,
    \end{cases}\]
for all $b_0,\ldots,b_{\ell-1}\in\{0,1\}$ and all $w_0,\ldots,w_{\ell-1}\in\{0,1\}^*$ such that $|w_0|=\ldots=|w_{\ell-1}|\geq \ell-1$. Then $d'$ is a $\bl$-martingale, and, for all $(w_0,\ldots,w_{n-1})\in\Phi$, the martingale value $d'(w_0,\ldots,w_{n-1})$ is at least a constant multiple of $d\left(\bigsqcup_{r=0}^{n-1} w_r\right)$. In particular, for every $(w_0,\ldots,w_{n-1})\in\Phi$ such that $(w_0,\ldots,w_{n-1})\sqsubseteq_\bl\b t$, the value $d'(w_0,\ldots,w_{n-1})$ is at least a constant multiple of $d(w)$, where $w$ is the length-$(n|w_0|)$ prefix of $z$. The value of $d$ on prefixes of $z$ is unbounded, and for every prefix $w'$ of $z$, we have $d(w')\leq 2^{n}d(w)$ for some prefix $w$ whose length is an integer multiple of $n$. Therefore, the value of $d'$ is unbounded on sequences $(w_0,\ldots,w_{n-1})\in\Phi$ such that $(w_0,\ldots,w_{n-1})\sqsubseteq_\bl\b t$, meaning $d'$ succeeds on $\b t$ and $\b t$ is not $\bl$-random.

To see that the converse does not hold, let $\bl\in[0,\infty)^\omega$ be any infinite rate sequence, and let $(t_0, t_1, \ldots)\in\b D_{\b \lambda}$ be such that $(t_0,\ldots,t_{n-1})$ is independently $(\lambda_0,\ldots,\lambda_{n-1})$-random for all $n\in\N$. Then $(t_0/2,\ldots,t_{n-1}/2)$ is also independently $(\lambda_0,\ldots,\lambda_{n-1})$-random for all $n\in\N$. There is a lower semicomputable $\bl$-martingale $d$ such that
\[d(0w_0,\ldots,0w_{n-1})=2^n\]
for all $n\in\N$ and all $w_1,\ldots,w_{n-1}\in\{0,1\}^*$. Intuitively, $d$ bets only on the first bit of each duration and hedges on all other bits. This martingale succeeds on $\b t=(t_0/2,\ldots,t_{n-1}/2)$ because $\bin(t_i/2)=0\bin(t_i)$. Therefore, $\b t$ is not $\b\lambda$-random.
\end{proof}

\section{Measure and martingales for continuous-time Markov chains}\label{sec:mmctmc}

\subsection{Continuous-time Markov chains}
A \emph{continuous-time Markov chain (CTMC)} is an ordered triple $$C = (Q, \lambda, \sigma),$$ where $Q$ is a countable set of states, $\lambda:Q \times Q \rightarrow [0,\infty)$ is a \emph{rate matrix} satisfying $\lambda(q,q) = 0$ for every $q \in Q$, and $\sigma: Q \rightarrow [0,1]$ is a \emph{state initialization} as described in section~\ref{sec:rss}.

Let $C = (Q, \lambda, \sigma)$ be a CTMC. At each time $t \in [0, \infty)$, $C$ is probabilistically in some state. At time $t = 0$, this state is chosen according to $\sigma$. For each state $q \in Q$, the real number $$\lambda_q = \sum_{r \in Q}\lambda(q, r)$$ is the \emph{rate out of} state $q$. If $\lambda_q = 0$, then $q$ is a \emph{terminal} state, meaning that if $C$ ever enters state $q$, then $C$ remains in state $q$ forever. If $\lambda_q > 0$, then $q$ is a \emph{nonterminal} state, and if $C$ enters $q$ at some time $t$, then the \emph{sojourn time} for which $C$ remains in state $q$ before moving to a new state is a random variable that has the exponential distribution with rate $\lambda_q$. Hence the expected sojourn time of $C$ in state $q$ is $\frac{1}{\lambda_q}$. When $C$ does move to a new state, it moves to state $r \in Q$ with probability
\[p(q, r) = \frac{\lambda(q,r)}{\lambda_q}.\]

Note that the CTMC model uses continuous time, with sojourn times ranging over $(0, \infty]$, but a discrete state space. Accordingly, its state transitions, called \emph{jump transitions}, are instantaneous.

\subsection{Trajectories in CTMCs}
A \emph{trajectory} of a CTMC $C = (Q, \lambda, \sigma)$ is a sequence $\b{\tau}$ of the form
\[\b{\tau} = ((q_n,t_n) \mid n \in \mathbb{N}) \in (Q \times (0,\infty])^\omega.\]
Intuitively, such a trajectory $\b{\tau}$ denotes the turn of events in which $q_0, q_1, \ldots $ are the successive states of $C$ and $t_0, t_1,\ldots $ are the successive sojourn times of $C$ in these states. Accordingly, we write
\[state_{\b{\tau}}(n) = q_n, \text{ } soj_{\b{\tau}}(n) = t_n\]
for each $n \in \mathbb{N}$.
When convenient, we write $\b{\tau}$ as an ordered pair
\[\b{\tau} = (\b{q}, \b{t}),\]
where $\b{q} = (q_n \mid n \in \mathbb{N})$ and $\b{t} = (t_n \mid n \in \mathbb{N})$.

There are two ways in which a trajectory $(\b{q}, \b{t})$ may fail to represent a ``true trajectory'' of the CTMC $C$ in the above intuitive sense. First, it may be  the case that $p(q_n, q_{n+1}) = 0$ (i.e., $\lambda(q_n, q_{n+1}) = 0$) for some $n \in \mathbb{N}$. This presents no real difficulty, since it merely says that the event ``$state_{\b{\tau}}(n) = q_n$ and $state_{\b{\tau}}(n+1) = q_{n+1}$'' has probability $0$. The second way in which  $(\b{q}, \b{t})$ may fail to represent a ``true trajectory'' is for some $q_n$ to be a terminal state of $C$. We deal with this by defining the \emph{length} of a trajectory $\b{\tau} = (\b{q}, \b{t})$ to be
\[\|\btau\| = \min\{n \in \mathbb{N} \mid q_n \text{ is terminal}\},\]
where $\min\emptyset = \infty$. We then intuitively interpret a trajectory $\b{\tau} =  (\b{q}, \b{t})$ with $\|\b{\tau}\| < \infty$ as the finite sequence $$\b{\tau}' = ((q_n,t_n') \mid n \leq \|\b{\tau}\|),$$ where each
\begin{equation*}
    t_n' =
    \begin{cases}
    t_n & \text{if } n < \|\b{\tau}\|\\
    \infty & \text{if } n = \|\b{\tau}\|.
    \end{cases}
\end{equation*}
We write $$\Omega = \Omega[C] = (Q \times (0,\infty])^\omega$$ for the set of all trajectories of $C$.

\subsection{Measure in CTMCs}

Elements of $(Q\times\{0,1\}^*)^*$ are called \emph{approximations} or \emph{partial specifications} of trajectories. The \emph{cylinder generated by} an approximation $w = ((q_0,u_0),\ldots,(q_{n-1},u_{n-1}))$ is the set
\[\Omega_w =\big\{ \b{\tau} \in \Omega \mid (\forall\, 0 \leq i < n)\ \big(state_{\b{\tau}}(i) = q_i\text{ and } soj_{\b{\tau}}(i) \in D_{\lambda_{q_i}}(u_i)\big)\big\};\]
in particular, the cylinder generated by the empty sequence $()$ is $\Omega_{()}=\Omega$. For two approximations $w,x$ and $\btau \in \Omega$, we write $w \pref \btau$ to indicate that $\btau \in \Omega_w$ and $w \pref x$ to indicate that $\Omega_x \subseteq \Omega_w$. It is easily verified that the latter condition is decidable. When it is convenient, we will write partial specifications in the form $(\b{q},\b{t})$, where $\b{q}\in Q^*$, $\b{t}\in(\{0,1\}^*)^*$, and $|\b{q}|=|\b{t}|$.

Let
\[\Psi=\Psi[C]=\bigcup_{n\in\N}(Q\times\{0,1\}^n)^n,\]
the set of all partial specifications which, for some $n$, specify $n$ states and $n$ sojourn times, each up to $n$ bits of precision. For each $w\in\Psi$, the \emph{probability} of the cylinder $\Omega_w$ in $C$, written $\mu_C(\Omega_w)$ or $\mu_C(w)$, is defined by $\mu_C(())=1$ and, for each $n\geq 1$ and $w=((q_0,u_0),\ldots,(q_{n-1},u_{n-1}))\in (Q\times\{0,1\}^n)^n$,
\[\mu_C(w)=\begin{cases}
    \sigma(q_0)\prod\limits_{i=0}^{\min\{n,\|w\|\}-2}\left(p(q_i,q_{i+1})2^{-|u_i|}\right)&\text{if }(\forall\, \|w\|\leq i<n)\ (q_i,u_i)=(q_{\|w\|},1^{n})\\
    0&\text{otherwise},
\end{cases}\]
where
\[\|w\|=\min\{i<n\mid q_i\text{ is terminal}\},\]
again with the convention $\min\emptyset=\infty$.

A set $X \subseteq \Omega$ has \emph{probability 0}, and we write $\mu_C(X)=0$, if there is a function $g:\N\times\N\to\Psi$ such that, for all $k\in\N$,
\[X\subseteq \bigcup_{\ell\in\N}\Omega_{g(k,\ell)}\]
and
\[\sum_{\ell\in\N}\mu_C(g(k,\ell))\leq 2^{-k}.\]

\subsection{Martingales for CTMCs}

A \emph{$C$-supermartingale} is a function $d:\Psi\to[0,\infty)$ such that, for all $w\in\Psi$,
\begin{equation}\label{eq:supermartingale}
d(w)\mu_C(w) \geq\sum_{\substack{w'\in\Psi\\w\sqsubseteq w'\\|w'|=|w|+1}}d(w')\mu_C(w').
\end{equation}
If equality holds in~\eqref{eq:supermartingale} for all $w\in\Psi$, then $d$ is a \emph{$C$-martingale}. A $C$-supermartingale $d$ \emph{succeeds} on a trajectory $\b\tau\in\Omega$, and we write $\b\tau\in S_C[d]$, if the set
\[\{d(w)\mid w\in\Psi\text{ and }w\sqsubseteq \b\tau\}\]
is unbounded.

\subsection{Analog of Ville's theorem for CTMC trajectories}

In order to prove an analog of Ville's theorem holds for CTMC trajectories, we first prove the following lemma, which is analogous to Kraft's inequality for prefix codes. A set $B$ of partial specifications of trajectories is a \emph{prefix-free set} if, for all distinct $w,x\in B$, we do not have $w\sqsubseteq x$.

\begin{lemma} [Kraft inequality for $C$-martingales]\label{lem:gki}
    Let $C = (Q, \lambda, \sigma)$ be a CTMC, $d$ a $C$-martingale, and $B \subseteq\Psi$  a prefix-free set. Then,
    \begin{equation}\label{eq:gki}
        d(())\geq \sum_{w\in B} d(w)\mu_C(w).
    \end{equation}
\end{lemma}

\begin{proof}
	If $d(()) = 0$, then~\eqref{eq:gki} is immediate. Hence, assume $d(()) >0$. Define the Borel probability measure $\pi:(Q\times\{0,1\}^\omega)^\omega\to[0,1]$ by
    \[\pi(w) = \frac{d(w)\mu_C(w)}{d(())}\]
    for all $w\in\Psi$. Now choose $\b\tau \in \Omega$ according to $\pi$, and let $E$ be the event that there exists some $w\in B$ such that $w\sqsubseteq \b\tau$. Then
    \begin{align*}
    1 &\geq \Pr(E)\\
    &= \sum_{w \in B} \pi(w)\\
    &= \frac{1}{d(())}\sum_{w \in B} d(w)\mu_C(w),
    \end{align*}
    which implies~\eqref{eq:gki}.
\end{proof}

We now prove the following analog of Ville's theorem~\cite{Ville39} for CTMC trajectories.
\begin{theorem}\label{thm:trajville}
    Let $C=(Q,\lambda,\sigma)$ be a CTMC and $X\subseteq\Omega[C]$. The following are equivalent. \begin{enumerate}
        \item $\mu_C(X)=0$
        \item There is a $C$-martingale $d$ such that $X\subseteq S_C[d]$.
        \item There is a $C$-supermartingale $d$ such that $X\subseteq S_C[d]$.
    \end{enumerate}
\end{theorem}
\begin{proof}
    (1 $\implies$ 2)\quad
    Suppose $\mu_{C}(X) = 0$, and let $g:\N\times\N\to\Psi$ be a function testifying to this. Assume without loss of generality that for all $k\in\N$, the collection $\{\Omega_{g(k,\ell)}\mid \ell\in\N\}$ of cylinders is pairwise disjoint. We wish to show that there exists a $C$-martingale $d$ such that $X \subseteq S_C[d]$.

    For each $k \in \N$, define the function $d_k: \Psi \to [0,\infty)$ by
    \[d_k(w) = |\{\ell \in \N \mid g(k,\ell) \sqsubseteq w\}|+ \sum_{\substack{\ell \in \N\\ w\sqsubset g(k,\ell)}} \frac{\mu_C(g(k,\ell))}{\mu_C(w)}.\]
    It is straightforward to verify that each $d_k$ is a $C$-martingale. Hence, the function $d:\Psi\to[0,\infty)$ defined by
    \[d(w) = \sum_{k \in \N} d_k(w)\]
    is also a $C$-martingale, which has initial capital
    \begin{align*}
        d(())&=\sum_{k\in\N}d_k(())\\
        &=\sum_{k\in\N}\left(|\{\ell \in \N \mid g(k,\ell) \sqsubseteq ()\}|+ \sum_{\substack{\ell \in \N\\ ()\sqsubset g(k,\ell)}} \frac{\mu_C(g(k,\ell))}{\mu_C(())}\right)\\
        &\leq \sum_{k\in\N}\left(|\{\ell \in \N \mid \mu_C(g(k,\ell))=1\}|+\sum_{\ell\in\N}\mu_C(g(k,\ell))\right)\\
        &\leq 1+\sum_{k\in\N}2^{-k}\\
        &=3,
    \end{align*}
    by the measure condition on $g$.

    To see that $X \subseteq S_C[d]$, let $\btau \in X$ and $\alpha \in \Z^+$. It suffices to show that there exists some $w\in\Psi$ such that $w\sqsubseteq \btau$ and $d(w) \geq \alpha$. For each $k\in\N$, define
    \[S^1[d_k] = \{w \in \Psi \mid d_k(w) \geq 1\},\]
    the \emph{unitary success set} of $d_k$.  Let $w_0,\ldots,w_{\alpha-1}$ be elements of $S^1[d_0],\ldots,S^1[d_{\alpha-1}]$, respectively, such that $w_k \sqsubseteq \btau$ for all $0 \leq k < \alpha$. These $w_k$ exist because, for all $k,\ell\in\N$,
    \[d_k(g(k,\ell))\geq |\{i\in\N\mid g(k,i)\sqsubseteq g(k,\ell)\}|\geq 1,\]
    so each $g(k,\ell)$ belongs to $S^1[d_k]$, and $\btau\in\bigcup_\ell g(k,\ell)$ holds for each $k\in\N$.
    
    Let $w$ be such that $w_k \sqsubseteq w$ for all $0 \leq k < \alpha$. Then
    \[d(w) = \sum_{k=0}^{\alpha-1} d_k(w) \geq \alpha,\]
    and we conclude that $X\subseteq S_C[d]$.

    (2 $\implies$ 1)\quad Now assume instead that there exists a $C$-martingale $d$ such that $X \subseteq S_C[d]$. Without loss of generality, assume $d$ has initial capital $d(())\leq 1$. We wish to show that $\mu_C(X) = 0$, i.e., that there exists a function $g:\N\times\N \to  \Psi$ satisfying, for all $k\in\N$, the covering property
    \[X\subseteq \bigcup_{\ell\in\N}\Omega_{g(k,\ell)}\]
    and the measure property
    \[\sum_{\ell\in\N}\mu_C(g(k,\ell))\leq 2^{-k}.\]

    For each $k \in \N$, define the set
    \[A_k = \left\{\btau\in\Omega \;\middle|\; \exists w\in\Psi\ \left(w\sqsubseteq \btau\text{ and }d(w) \geq 2^k\right)\right\}.\]
    The $A_k$ are open sets, so there is a function $g:\N\times\N\to\Psi$ such that, for all $k\in\N$, the set $\{g(k,\ell)\mid \ell\in\N\}$ is prefix-free and
    \[A_k=\bigcup_{\ell\in\N}\Omega_{g(k,\ell)},\]
    This union is disjoint by the prefix-free condition, so it suffices to show that $X\subseteq A_k$ and $\mu_C(A_k)\leq 2^{-k}$.

    Fix $k\in\N$. Since $X\subseteq S_C[d]$, we have that for all $\btau \in X$, there exists $w \in \Psi$ such that $w \sqsubseteq \btau$ and $d(w) > 2^{k+1}$, so $X\subseteq A_k$. To see that $\mu_C(A_k)\leq 2^{-k}$, let $B_k\subseteq \Psi$ be a set that includes, for each $\btau\in A_k$, exactly one prefix $w\sqsubset\btau$ such that $d(w)\geq 2^{k}$. Then
    \begin{align*}
        \mu_C(A_k)&=\mu_C\left(\bigcup_{w\in B_k}\Omega_w\right)\\
        &=\sum_{w\in B_k}\mu_C(w)\\
        &\leq 2^{-k}\sum_{w\in B_k}d(w)\mu_C(w)\\
        &\leq 2^{-k},
    \end{align*}
    by Lemma~\ref{lem:gki}. We conclude that $\mu_C(X)=0$.

    (2 $\iff$ 3)\quad The second condition trivially implies the third. For the reverse direction, observe that, given a $C$-supermartingale in which~\eqref{eq:supermartingale} is strict for some $w\in\Psi$, we must have $\mu_C(w)>0$, and therefore there is some $w'\in\Psi$ such that $w\sqsubseteq w'$, $|w'|=|w|+1$, and $\mu_C(w')>0$. Hence, we can increase $d(w')$ for the lexicographically first, shortest such $w'$ to make equality hold in~\eqref{eq:supermartingale} for $w$ while maintaining~\eqref{eq:supermartingale} on all of $\Psi$ and not introducing new violations of equality in~\eqref{eq:supermartingale}, except potentially at $w'$. We can repeat this alteration of $d$ for all $n\in\N$, for all
    \[w\in((Q\upharpoonright n)\times\{0,1\}^{\leq n})^{\leq n}\subseteq\Psi,\]
    where $Q\upharpoonright n$ is the first $n$ states in the lexicographic enumeration of $Q$. Let $d'$ be the result of this infinite sequence of alterations. As all violations of equality in~\eqref{eq:supermartingale} have been resolved, $d'$ is a $C$-martingale such that $d'(w)\geq d(w)$ for all $w\in\Psi$. Therefore $S_C[d]\subseteq S_C[d']$, and the third condition implies the second.
\end{proof}

A slightly simplified version of the above proof yields the following analog of Ville's theorem \cite{Ville39} for sojourn time sequences.

\begin{theorem}
\label{ville_rate_seq}
If $\bl$ is a rate sequence, then, for each set $E \subseteq \b{D_\bl},$ the following two conditions are equivalent.
\begin{enumerate}
    \item $\mu_\bl(E) = 0$.
    \item There is a $\bl$-martingale $d$ such that $E \subseteq S^\infty[d]$.
\end{enumerate}
\end{theorem}

\section{Random CTMC trajectories}\label{sec:rctmct}

\subsection{Algorithmic measure and randomness in CTMCs}
Let $C=(Q,\lambda,\sigma)$ be a CTMC. We assume that the states $q\in Q$ have canonical representations, so that it is clear what it means for functions $f:Q\rightarrow Q$ etc. to be computable.

A set $X \subseteq \Omega[C] $ has \emph{algorithmic probability 0}, and we write $\mu_{C,\constr}(X)=0$, if there is a computable function $g:\N\times\N\to\Psi$ such that, for all $k\in\N$,
\[X\subseteq \bigcup_{\ell\in\N}g(k,\ell)\]
and
\[\sum_{\ell\in\N}\mu_C(g(k,\ell))\leq 2^{-k}.\]
We define an individual trajectory $\btau$ of a CTMC $C$ to be  \emph{algorithmically random} if there is no lower semicomputable $C$-martingale $d$ such that $\btau\in S_C[d]$.

\subsection{Analog of Schnorr's theorem for CTMCs}

\begin{theorem}\label{thm:trajschnorr}
    Let $C=(Q,\lambda,\sigma)$ be a CTMC such that $\lambda$ and $\sigma$ are computable, and let $X\subseteq\Omega[C]$. The following are equivalent.
    \begin{enumerate}
        \item $\mu_{C,\constr}(X)=0$.
        \item There is a lower semicomputable $C$-martingale $d$ such that $X\subseteq S_C[d]$.
    \end{enumerate}
\end{theorem}
\begin{proof}
    It suffices to observe that the arguments for (1$\implies$2) and (2$\implies$1) in the proof of Theorem~\ref{thm:trajville} can be effectivized. If $g$ is a computable function testifying to $\mu_{C,alg}(X)=0$, then each $d_k$ is lower semicomputable, so their sum $d$ is a lower semicomputable $C$-martingale that succeeds on $X$. Conversely, if $d$ is a lower semicomputable $C$-martingale that succeeds on $X$, then we can define a computably enumerable $B_k$ by including, for each $\btau\in A_k$, the first-discovered prefix $w\sqsubseteq \btau$ such that $d(w)\geq 2^{k+1}$. The resulting $g$ is a computable function testifying to $\mu_{C,\constr}(X)=0$.
\end{proof}

A slightly simplified version of the above proof yields the following analog of Schnorr's theorem \cite{DBLP:journals/mst/Schnorr71} for sequences of sojourn times.

\begin{theorem}\label{schnorr_rate_seq} If $\bl$ is a computable rate sequence, then, for each set $E \subseteq \b{D_{\lambda}}$, the following two conditions are equivalent.
\begin{enumerate}
    \item $\mu_{\bl,\constr}(E) = 0$.
    \item There is a lower semicomputable $\bl$-martingale $d$ such that $E \subseteq S^\infty[d]$.
\end{enumerate}
\end{theorem}   

\subsection{Analog of van Lambalgen's theorem for CTMC trajectories}

\begin{theorem}
    Let $C=(Q,\lambda,\sigma)$ be a CTMC such that $\lambda$ and $\sigma$ are computable, and let $\b{q}=(q_n\mid n\in\N)\in Q^\omega$, $\b{t}=(t_n\mid n\in\N)\in(0,\infty]^\omega$, and $\bl=(\lambda_{q_n}\mid n\in\N)\in [0,\infty)^\omega$. The trajectory $(\b{q},\b{t})$ is random if and only if $\b{q}$ is random and $\b{t}$ is $\bl$-random relative to an oracle for $\b{q}$.
\end{theorem}
\begin{proof}
    First suppose $\b{q}$ is not random. Then by Theorem~\ref{schnorr_state}, we have $\mu_{\sq,\sigma,\constr}(\{\b{q}\})=0$, where $\sq$ is the probabilistic transition system for the states of $C$, with initialization $\sigma$. Hence, there is a computable function $g:\N\times\N\to Adm_\sq(\sigma)$ such that, for all $k\in\N$,
    \[\sum_{j\in\N}\mu_{\sq,\sigma}(g(k,j))\leq 2^{-k}\]
    and there is some $\ell$ such that $g(k,\ell)\sqsubseteq \b{q}$.
    
    Define a computable function $g':\N\times\N\to\Psi$ as follows. Fix $k\in\N$. For each $j\in\N$, let $n=|g(k+1,j)|$ and $W_{k,j}=(\{0,1\}^{n})^{n}$. Lexicographically list the $2^{n^2}$ elements $((r_0,w_0),\ldots,(r_{n-1},w_{n-1}))$ of $\Psi$ such that $(r_0,\ldots,r_{n-1})=g(k+1,j)$ and $(w_0,\ldots,w_{n-1})\in W_{k,j}$, and let $A_k$ be the concatenation, over all $j\in\N$, of these lists. For each $\ell\in\N$, if $\ell\leq |A_k|$, define $g'(k,\ell)$ to be the $\ell$\textsuperscript{th} entry in $A_k$; otherwise, define $g'(k,\ell)=(\b{r},(0^{k+\ell+2})^{k+\ell+2})$, where $\b{r}\in Q^{k+\ell+2}$ is arbitrary, so that $\mu_C(g'(k,\ell))\leq 2^{-(k+\ell+2)^2}<2^{-k-\ell-2}$ whenever $\ell>|A_k|$.
    
    Then for each $k\in\N$, as there is some $\ell\in\N$ such that $g(k+1,\ell)\sqsubseteq \b{q}$, there is also some $\ell'\in\N$ such that $g'(k,\ell')\sqsubseteq(\b{q},\b{t})$, and
    \begin{align*}
        \sum_{\ell\in\N}\mu_C(g'(k,\ell))&=\sum_{\ell\leq|A_k|}\mu_C(g'(k,\ell))+\sum_{\ell>|A_k|}\mu_C(g'(k,\ell))\\
        &<\sum_{j\in\N}\sum_{\b{w}\in W_{k,j}}\mu_{\sq,\sigma}(g(k+1,j))\mu_\bl(\b{w})+\sum_{\ell\in\N}2^{-k-\ell-2}\\
        &=\sum_{j\in\N}\left(\mu_{\sq,\sigma}(g(k+1,j))\sum_{\b{w}\in W_{k,j}}\mu_\bl(\b{w})\right)+2^{-k-1}\\
        &=\sum_{j\in\N}\mu_{\sq,\sigma}(g(k+1,j))+2^{-k-1}\\
        &\leq 2^{-k}.
    \end{align*}
    Thus, $g'$ testifies to $\mu_{C,\constr}(\{(\b{q},\b{t})\})=0$. By Theorem~\ref{thm:trajschnorr}, we conclude that $(\b{q},\b{t})$ is not random in this case.

    Suppose instead that $\b{t}$ is not $\bl$-random relative to an oracle for $\b{q}$. Then there is an oracle Turing machine $M$ where $M^{\b{q}}$ computes a function $g^{\b{q}}:\N\times\N\to\Phi$ such that, for every $k\in\N$,
    \[\sum_{j\in\N}\mu_\bl(g^{\b{q}}(k,j))\leq 2^{-k}\]
    and there is some $j\in\N$ such that $g^{\b{q}}(k,j)\sqsubseteq\b{t}$. We can modify this $M$ so that it also has the property that for \emph{every} oracle $A$, $M^{A}$ computes a function $g^A:\N\times\N\to\Phi$ such that, for every $k\in\N$,
    \[\sum_{j\in\N}\mu_\bl(g^A(k,j))\leq 2^{-k}.\]

    For each $k,m\in\N$, let
    \[B_{k,m}=\{(\b{r},\b{u})\in\Psi\mid |\b{r}|=m\text{ and }\exists j\in \N\ (g^{\b{r}}(k,j)\sqsubseteq \b{u})\}\]
    and
    \[V_{k,m}=\bigcup_{(\b{r},\b{u})\in B_{k,m}}\Omega_{(\b{r},\b{u})}.\]
    Then
    \begin{align*}
        \mu_C(V_{k,m})&\leq \sum_{\b{r}\in Q^m}\sum_{j\in\N}\mu_{\sq,\sigma}(\b{r})\mu_{\bl}(g^{\b{r}}(k,j))\\
        &=\sum_{\b{r}\in Q^m}\mu_{\sq,\sigma}(\b{r})\sum_{j\in\N}\mu_{\bl}(g^{\b{r}}(k,j))\\
        &\leq 2^{-k}\sum_{\b{r}\in Q^m}\mu_{\sq,\sigma}(\b{r})\\
        &=2^{-k}.
    \end{align*}
    Noting that $V_{k,m}\subseteq V_{k,m+1}$ and letting
    \[V_k=\bigcup_{m\in\N}V_{k,m},\]
    this implies $\mu_C(V_k)\leq 2^{-k}$. Furthermore, as there is some $j\in\N$ such that $g^{\b{q}}(k,j)\sqsubseteq\b{t}$, we must have $(\b{q},\b{t})\in V_k$. Since the $V_k$ are uniformly computably enumerable open sets, there is therefore a computable function $g':\N\times\N\to\Psi$ such that, for each $k\in\N$,
    \[\bigcup_{\ell\in\N}\Omega_{g'(k,\ell)}=V_k\]
    and
    \[\sum_{\ell\in\N}\mu_C(g'(k,\ell))\leq 2^{-k}.\]
    This function $g'$ testifies to $\mu_{C,\constr}(\{(\b{q},\b{t})\})=0$, and we conclude, again by Theorem~\ref{thm:trajschnorr}, that $(\b{q},\b{t})$ is not random in this case.

    Conversely, suppose $(\b{q},\b{t})$ is not random. For this direction, we follow the proof in~\cite{downey2010algorithmic}, attributed to Nies, of part of van Lambalgen's theorem. Since $(\b{q},\b{t})$ is not random, Theorem~\ref{thm:trajschnorr} tells us there is some computable function $g:\N\times\N\to\Psi$ testifying to $\mu_{C,\constr}(\{(\b{q},\b{t})\})=0$. For all $k\in\N$, let
    \[V_k=\bigcup_{\ell\in\N} \Omega_{g(2k+1,\ell)},\]
    so the $V_k$ are uniformly computably enumerable and, for all $k\in\N$, we have $(\b{q},\b{t})\in V_k$ and
    \begin{equation}\label{eq:vk}
        \mu_C(V_k)\leq\sum_{\ell\in\N}\mu_C(g(2k+1,\ell))
        \leq 2^{-2k-1}.
    \end{equation}

    For all $k\in\N$, using the notation~\eqref{eq:Qcylinder}, define the sets
    \[A_k=\left\{\b{s}\in Q^*\mid \mu_C(V_k\cap(\mathds{A}_{\b{s}}(\sigma)\times\Phi))>2^{-k}\mu_{\sq,\sigma}(\b{s})\right\},\]
    \[S_k=\bigcup_{\b{s}\in A_k}\mathds{A}_{\b{s}}(\sigma),\]
    and
    \[S'_k=\bigcup_{j\geq k}S_j.\]
    Thus, $S'_{k+1}\subseteq S'_k$ for all $k\in\N$, and the $S_k$ and $S'_k$ are uniformly computably enumerable. For all $k\in\N$, let $A'_k\subseteq A_k$ be a prefix-free set such that
    \[S_k=\bigcup_{\b{s}\in A'_k}\mathds{A}_{\b{s}}(\sigma).\]
    Then for all $k\in\N$,
    \begin{align*}
        \mu_{\sq,\sigma}(S_k)&\leq\sum_{\b{s}\in A'_k}\mu_{\sq,\sigma}(\b{s})\\
        &<2^k\sum_{\b{s}\in A'_k}\mu_C(V_k\cap(\mathds{A}_{\b{s}}(\sigma)\times\Phi))\tag{$A'_k\subseteq A_k$}\\
        &\leq 2^k\mu_C(V_k)\tag{$A'_k$ is prefix-free}\\
        &\leq 2^{-k-1}\tag{by~\eqref{eq:vk}}.
    \end{align*}
    Therefore,
    \[\mu_{\sq,\sigma}(S'_k)\leq \sum_{j\geq k}\mu_{\sq,\sigma}(S_j)<\sum_{j\geq k}2^{-j-1}=2^{-k}.\]
    Since the $S'_k$ are uniformly computably enumerable open sets, there is a computable function $g:\N\times\N\to Adm_\sq(\sigma)$ such that, for each $k\in\N$, the cylinders in $\left\{\mathds{A}_{g(k,\ell)}(\sigma)\;\middle|\; \ell\in\N\right\}$ are pairwise disjoint, and their union is $S'_k$. Suppose that for all $k\in\N$ we have $\b{q}\in S'_k$. Then this function $g$ testifies to $\mu_{\sq,\sigma,\constr}(\{\b{q}\})=0$, so Theorem~\ref{schnorr_state} implies that $\b{q}$ is not random in this case.

    Hence, assume there is some $k_0\in\N$ such that $\b{q}\not\in S'_{k_0}$, noting that this implies $\b{q}\not\in S_k$ for all $k\geq k_0$. For all $i\in\N$, let $\b{s}_i=\b{q}\upharpoonright (i+1)$, which cannot be in $A_k$ for any $k\geq k_0$. Since $\mu_{\sq,\sigma}(\b{s}_i)= 0$ would imply (again by Theorem~\ref{schnorr_state}) that $\b{q}$ is not random, further assume $\mu_{\sq,\sigma}(\b{s}_i)> 0$. For all $i,k\in\N$, define the sets
    \[B_k^i=\{\b{v}\in\Phi\mid |\b{v}|=i\text{ and }\mathds{A}_{\b{s}_i}(\sigma)\times\b{D_\lambda}(\b{v})\subseteq V_k\},\]
    and
    \[T_k^i=\bigcup_{\b{v}\in B_k^i}\b{D_\lambda}(\b{v}),\]
    noting that $T_k^i\subseteq T_k^{i+1}$. For all $k\geq k_0$, we have
    \begin{align*}
        \mu_{\sq,\sigma}(\b{s}_i)\mu_{\bl}(T_k^i)&=\mu_C(\mathds{A}_{\b{s}_i}(\sigma)\times T_k^i)\\
        &\leq \mu_C(V_k\cap (\mathds{A}_{\b{s}_i}(\sigma)\times\Phi))\\
        &\leq 2^{-k}\mu_{\sq,\sigma}(\b{s}_i),
    \end{align*}
    where the last inequality holds because $\b{s}_i\not\in A_k$. Therefore, $\mu_{\bl}(T_k^i)\leq 2^{-k}$.
    
    For each $k\in\N$, define the set $T_k=\bigcup_{i\in\N} T_k^i$. Then $\b{t}\in T_k$ for all $k\in\N$, and for all $k\geq k_0$,
    \[\mu_{\bl}(T_k)=\sup_{i\in\N}\mu_{\bl}(T_k^i)\leq 2^{-k}.\]
    Since the $T_k$ are open sets and are uniformly computably enumerable relative to $\b{q}$, there is a function $g:\N\times\N\to\Phi$ that is computable relative to $\b{q}$ such that, for each $k\in\N$, the cylinders in $\left\{\b{D}_\bl(g(k,\ell))\;\middle|\;\ell\in\N\right\}$ are pairwise disjoint, and their union equals $T_k$. This $g$ is a witness to $\b{t}$ having algorithmic probability 0 relative to $\b{q}$, and it follows by Theorem~\ref{schnorr_rate_seq} that $\b{t}$ is not $\bl$-random relative to $\b{q}$ in this case. 
\end{proof}

\section{Kolmogorov complexity and CTMC randomness}

Random trajectories can also be characterized using \textit{Kolmogorov complexity}. First, we briefly review this notion in the classical setting. We fix a universal \textit{self-delimiting Turing machine} (see \cite{oLiVit19}), $U$. The Kolmogorov complexity, $K$, of a (finite) string $x$ in $\{0,1\}^*$ is the length of a shortest program for a self-delimiting Turing machine which prints $x$. That is, $K: \{0,1\}^* \rightarrow \mathbb{N}$ is defined by $$K(x) = \min \{|\pi| \mid U(\pi) = x\ \text{ and } \pi \in \{0,1\}^*\}.$$ When $x$ is not a binary string, but some other finite object, $K(x)$ is defined from the above by routine coding.

In the terminology of~\cite{chaitin1987incompleteness}, an \emph{information content measure} is a function $f:\N\to\N$ that is upper semicomputable and satisfies $\sum_{n\in\N}2^{-n}<1$. Will repeatedly use the fact that $K$ is a minimal information content measure, in the sense that for every information content measure $f$, there is a constant $c$ such that, for all $i\in\N$, we have $K(n)\leq f(n)+c$.

In this section, we deal with the full class $(Q\times\{0,1\}^*)^*$ of partial specifications of trajectories rather than just the restricted class $\Psi$ that we used to define measure and randomness in sections~\ref{sec:mmctmc} and~\ref{sec:rctmct}. This allows a partial specification to approximate different sojourn times to different levels of precision, which lets us state the following results in a more general way, although these theorems also hold for the restricted class $\Psi$. The \textit{profile} of a cylinder $\Omega_w$, where $w=((q_0,u_0),\ldots,(q_{n-1},u_{n-1}))\in (Q\times\{0,1\}^*)^n$, is an element of $\N^n$, denoted
\[\prof(w) = (|u_0|,\ldots ,|u_{n-1}|).\]
The profile of the empty sequence is $\prof(())=()$.

\begin{observation}
    For each CTMC $C$ and each profile $p$,
	\[\sum_{w\mid \prof(w)=p}\mu_C(w) =1.\]
\end{observation}

For each partial specification $w\in(Q\times\{0,1\}^*)^*$ of a CTMC $C$, the \emph{self-information} of $w$ is
\[h(w)=\log\left(\frac{1}{\mu_C(w)}\right).\]
The following two lemmas are analogous to standard results used in the Kolmogorov complexity characterization of algorithmically random sequences. We let $p$ range over all profiles, and assume some natural encoding (enumerating process) between natural numbers and profiles, and also between natural numbers and cylinders. 

\begin{lemma}\label{helperLemma}
	In a CTMC $C$, every partial specification $w$ satisfies
	\[K(w)\leq h(w)+ K(\prof(w))+ O(1).\]
    In particular, if $w\in\Psi$ then $K(w)\leq h(w)+K(|w|)+O(1)$.
\end{lemma}
\begin{proof}
    We have
    \begin{align*}
         1 &\geq \sum_{p} 2^{-K(p)} \\
         &= \sum_{p} \left(2^{-K(p)}\sum_{\prof(w)=p}2^{-h(w)}\right) \\
         &= \sum_p \sum_{\prof(w) = p} 2^{-(K(\prof(w)) + h(w))}\\
         &= \sum_{w} 2^{-(K(\prof(w))+h(w))}
    \end{align*}
    where the inequality is due to the minimality of $K$ as an information content measure, and the first equality holds because the inner sum equals $1$. That is, the upper semicomputable function
    \[w\mapsto K(\prof(w))+h(w)\]
    is an information content measure, and it follows, again by the minimality of $K$, that
    \[K(w)\leq  h(w)+ K(\prof(w))+ O(1).\]
    The inequality holds for $w\in\Psi$ because in that case $\prof(w)=(|w|)^{|w|}$, which is computable given $|w|$.
\end{proof}

When $W$ is a set of partial specifications, we write $\mu_C(W)$ as shorthand for $\mu_C\big(\bigcup_{w\in W}\Omega_w\big)$.

\begin{lemma}\label{lem:mainLem}
    There is a constant $c \in \mathbb{N}$ such that, for every profile $p$ of a CTMC $C$ and every $k \in \mathbb{N}$,
    \[\mu_C(\{w\mid \prof(w) = p\text{ and }K(w) < h(w) + K(p) -k\}\big)<2^{c-k}.\]
    In particular, for every $k,\ell\in\N$,
    \[\mu_C(\{w\mid w\in\Psi,\ |w| = \ell,\text{ and }K(w) < h(w) + K(\ell) -k\}\big)<2^{c-k}.\]
\end{lemma}

\begin{proof}
    We note that
    \[
     \sum_{p}\sum_{\prof(w)=p} 2^{-K(w)}=\sum_{w}2^{-K(w)}<1.
    \]
    Defining $f:\N^*\to\R$ by
    \[f(p) = \sum_{\prof(w)=p} 2^{-K(w)},\]
    the above says 
    \[\sum_{p}2^{-(-\log f(p))}  = \sum_{p}f(p)< 1.\]
    So $-\log(f(p))$ is an information content measure. 
    Then by the minimality of $K$ as an information content measure~, there exists a constant $c$ such that
    \[
        K(p)\leq -\log(f(p)) + c,
    \]
    and hence 
    \[ -K(p)+c \geq \log(f(p)). \]
    Exponentiating both sides and applying the definition of $f(p)$ yields
    \begin{align*}
      2^{-K(p)+c}&\geq \sum_{\prof(w)=p} 2^{-K(w)} \\
                 & =\sum_{\prof(w)=p} \mu_C(w)\frac{1}{\mu_C(w)}2^{-K(w)}\\
                 &=\sum_{\prof(w)=p} \mu_C(w)2^{\log{\frac{1}{\mu_C(w)}}}2^{-K(w)}\\
                 &= \sum_{\prof(w)=p} \mu_C(w)2^{h(w)-K(w)}\\
                 &=\mathbb{E}_{\mu_C}[2^{h(w)-K(w)}\mid \prof(w)=p].
    \end{align*}
    Therefore, letting $\mu'_C=\mu_C(\,\cdot\mid \prof(w)=p)$,
    \begin{align*}
        \mu'_C\left(\left\{w\mid K(w)<h(w)+K(\prof(w))-k\right\}\right)
        &=\mu'_C\left(\left\{w\mid h(w)-K(w)>k-K(\prof(w))\right\}\right)\\
        &=\mu'_C\left(\left\{w\;\middle|\; 2^{h(w)-K(w)}>2^{k-K(\prof(w))}\right\}\right)\\
        &<\frac{\mathbb{E}_{\mu'_C}[2^{h(w)-K(w)}]}{2^{k-K(\prof(w))}}\\
        &\leq \frac{2^{-K(\prof(w))+c}}{2^{k-K(\prof(w))}}\\
        &=2^{c-k},
    \end{align*}
where the first inequality holds by the Markov inequality and the second holds by Lemma~\ref{helperLemma}.
\end{proof}

With these lemmas, we can establish the Kolmogorov complexity characterization of randomness for trajectory objects, which is exactly analogous to a well-known characterization of the algorithmic randomness of sequences over finite alphabets \cite{zvonkin1970complexity,schnorr1977survey}.

\begin{theorem} \label{thmKolChar}
Let $C$ be any CTMC such that $\mu_C$ is a computable measure. Then a trajectory $\btau$ in $C$ is random if and only if there exists $k \in \mathbb{N}$ such that for every $w \pref \btau$,
$K(w)\geq h(w)-k$.
\end{theorem}
\begin{proof}
    First, assume that for every $k$, there is at least one $w\pref \btau$, such that $K(w)< h(w)-k$. It suffices to show that $\btau$ is not random. We let
    \[U_k=\{w\mid K(w)< h(w)-k\}.\]
    By our assumption,
    \[\btau\in\bigcap_{k\in\N}\bigcup_{w\in U_k}\Omega_w.\]
    For each profile $p\in\N^*$, define the \emph{$p$-slice} of $U_k$ to be
    \[U_k^p=\{w\mid w\in U_k \text{ and } \prof(w)=p \}.\]
    Note that by Lemma \ref{lem:mainLem}, we have $\mu_C(U_k^p)<2^{c-k-K(p)}$. Therefore,
    \begin{align*}
        \mu_C(U_k)&\leq\sum_p \mu_C(U_k^p)\\
        &<\sum_p 2^{c-k-K(p)}\\
        &\leq 2^{c-k},
    \end{align*}
    by the Kraft inequality. As the sets $U_0, U_1, \ldots$ are uniformly computably enumerable, it follows that
    \[\bigcap_{k\in\N}\bigcup_{w\in U_{c+k}}\Omega_w\]
    has algorithmic probability 0 and contains $\btau$. Thus, by Theorem~\ref{thm:trajschnorr}, $\btau$ is not random.
    
    Conversely, assume that $\btau$ is not random. Then, by Theorem~\ref{thm:trajschnorr}, there exists a computable function
    \[g:\N\times\N\to\Psi[C]\]
    such that, for every $k\in \mathbb{N}$,
    \[\btau \in \bigcup_{\ell\in\N}\Omega_{g(k,\ell)}\]
    and
    \[\sum_{\ell\in\N}\mu_C(g(k,\ell))\leq 2^{-k}.\]
    
    Define the function $F:\N\times\N\to\N$ by $F(k,\ell)=h(g(2k+1,\ell))-k$ and the function $f:\N\to\N$ by $f(i)=F(\pi^{-1}(i))$, where $\pi:\N\times\N\to\N$ is the Cantor pairing function. Then
    \begin{align*}
        \sum_{i\in\N}2^{-f(i)}&=\sum_{k\in\N}\sum_{\ell\in\N}2^{-F(k,\ell)}\\
        &=\sum_{k\in\N}2^k\sum_{\ell\in\N}2^{-h(g(2k+1,\ell))}\\
        &=\sum_{k\in\N}2^k\sum_{\ell\in\N}\mu_C(g(2k+1,\ell))\\
        &\leq \sum_{k\in\N}2^{-k-1}\\
        &=1.
    \end{align*}
    As $\mu_C$ is a computable measure, the set $\{(i,m))\mid f(i)\leq m\}$ is computably enumerable, so we have shown that $f$ is an information content measure. By the minimality of $K$ as an information content measure, for all $k,\ell\in\N$,
    \begin{align*}
        K(g(k,\ell))&\leq K(\pi(k,\ell))+O(1)\\
        &\leq f(\pi(k,\ell))+O(1)\\
        &\leq F(k,\ell)+O(1)\\
        &=h(g(2k+1,\ell))-k+O(1).
    \end{align*}
    Since for all $k\in\N$, there is some $\ell\in\N$ such that $g(2k+1,\ell)\sqsubseteq\btau$, we conclude there is no $k\in\N$ such that for every $w \pref \btau$, $K(w)\geq h(w)-k$.
\end{proof}

\section{Chemical reaction network randomness and Zeno phenomena}

In this last section we use our machinery to prove a fundamental fact about random trajectories in stochastic mass-action chemical reaction networks. We first briefly review the varieties of such chemical reaction networks and indicate how they are all special cases of CTMCs. More extensive discussions of this appear in~\cite{jSoCoWiBr08,oCSWB09,oFeinberg19} and elsewhere.

A \emph{chemical reaction network} (\emph{CRN}) is a pair $N = (S, R)$, where $S$ is a finite set of \emph{species} (intuitively, types of molecules), and $R$ is a finite set of \emph{reactions}, each of which is a triple $\rho = (r, p, k) \in \N^S \times \N^S \times [0, \infty)$, where $r$ and $p$ are as in the rate-free CRNs defined in section~\ref{sec:bts} above (recalling that $p$ and $r$ are distinct), and $k$ is the \emph{rate constant} of the reaction $\rho$.

Stochastic mass-action CRNs are usually just called stochastic CRNs, because mass-action kinetics are so commonly used as to be the default. Even so, there are several inequivalent definitions. These definitions have the following commonalities.
\begin{enumerate}[(i)]
    \item A \emph{state} of a stochastic CRN $N = (S, R)$ is a vector $q \in \N^S$ whose intuitive meaning is that, for each $X \in S$, there are $q(X)$ molecules of the species $X$ in solution in the state $q$. That is, $q(X)$ is the \emph{molecular count} of the species $X$ in the state $q$.
    \item The \emph{rate} $\lambda(\rho, q)$ of a reaction $\rho = (r, p, k)$ of $N$ in a state $q$ is determined by $\rho$, $q$, and perhaps the volume $V$ of the solution, which is taken to be a constant. In all the definitions,
    \begin{equation}\label{eq:crn1}
        \lambda(\rho, q) = k u P,
    \end{equation}
    where the factor $u$ is bounded by a constant and
    \begin{equation}\label{eq:crn2}
        P = \prod_{X \in S} q(X)^{r(X)}.
    \end{equation}
    It is the form of the product in~\eqref{eq:crn2} that characterizes the kinetics here as ``mass-action''. The factor $u$ above incorporates two variabilities in the literature, namely, whether and how the volume $V$ is incorporated and whether the ``sampling with replacement'' factors $q(X)^{r(X)}$ are replaced by the corresponding ``sampling without replacement'' factors, which are potentially smaller but more realistic in nanoscale applications with small molecular counts.
    \item Given a stochastic CRN $N = (S, R)$ as above, together with a probabilistic initialization $\sigma: \N^S \to [0,1]$ defined as in section~\ref{sec:rss} (and noting that this initialization is deterministic if some $\sigma(q) = 1$), the \emph{CTMC of $N$ initialized by $\sigma$} is the CTMC
    \[C(N, \sigma) = (\N^S, \hat\lambda, \sigma),\]
    where
    \[\hat\lambda: \N^S \times \N^S \to [0, \infty)\]
     is defined by
    \[\hat\lambda(q, r) = \sum_{\substack{\rho \in R\\\rho\text{ transforms }q\text{ to }r}} \lambda(\rho, q)\]
    for all $q, r \in \N^S$.
\end{enumerate}

The behavior of the stochastic CRN $N$ with the initialization $\sigma$ is \emph{by definition} the behavior of the CTMC $C(N, \sigma)$ defined as above.

By~\eqref{eq:crn1} and~\eqref{eq:crn2}, we have the following.
\begin{observation}\label{obs:bddrate}
    If the molecular counts along a trajectory $\btau$ of a stochastic CRN $N$ are bounded, then the rates of the reactions along $\btau$ are also bounded.
\end{observation}

We now prove our final result.
\begin{theorem}[Non-Zeno property]\label{non_zeno}
  Let $N$ be a stochastic CRN and $\sigma$ an initialization of $N$. If
  \[\b\tau = ((q_0, t_0),(q_1, t_1), \ldots  )\in \Omega[C(N,\sigma)]\]
  is random and has bounded molecular counts, then $\b\tau$ satisfies
  \begin{equation}\label{eq:non-zeno}
      \sum_{n=0}^\infty t_i = \infty.
  \end{equation}
\end{theorem}
\begin{proof}
    Let $\b{\tau} =((q_0,t_0),(q_1,t_1),\ldots)\in \Omega[C(N,\sigma)]$ be a trajectory with bounded molecular counts. Let
    \[M=\sup_{i\in\N}\hat\lambda_{q_i},\]
    which is finite by Observation~\ref{obs:bddrate}. Then for all $i\in\N$,
    \begin{align*}
        \b{D}_{\hat\lambda_{q_i}}(0)&=(0,F_{\hat\lambda_{q_i}}^{-1}(1/2)]\\
        &=(0,\ln(2)/\hat\lambda_{q_i}]\\
        &\supseteq(0,\ln(2)/M].
    \end{align*}
    
    Let $\Psi=\Psi[C(N,\sigma)]$, and suppose~\eqref{eq:non-zeno} does not hold. Then only finitely many of the $t_i$ can be greater than $\ln(2)/M$, so there must exist $k \in \N$ such that, for all $i \geq k$, we have $t_i \in \b{D}_{\hat\lambda_{q_i}}(0)$. Hence, for every $w=((r_0,u_0),\ldots,(r_{n-1},u_{n-1}))\in\Psi$ such that $w\sqsubseteq\btau$ and $n>k$, we have $u_k[0]=\ldots=u_{n-1}[0]=0$, where $u[0]$ denotes the first bit of a string $u$.
    
    Define a $C$-martingale $d:\Psi\to[0,\infty)$ by $d(())=1$ and, for all $w=((r_0,u_0),\ldots,(r_{n-1},u_{n-1}))\in\Psi$ and all $w'=((r_0,u'_0),\ldots,(r_n,u'_n))\in\Psi$ such that $w\sqsubseteq w'$,
    \[d(w') =
        \begin{cases}
        d(w) & \text{ if } n < k\\
        2 d(w) & \text{ if } n \geq k\text{ and } u'_n[0] = 0\\
        0 & \text{ if } n \geq k\text{ and } u'_n[0] = 1.
        \end{cases}
    \]
    Intuitively, $d$ does not begin to bet until it reaches the $(k+1)$\textsuperscript{th} sojourn time, and $d$ bets only on the first bit of each subsequent sojourn time. Thus, $d$ succeeds on $\b\tau$, and $d$ is clearly lower semicomputable. Therefore, $\b\tau$ cannot be random.
\end{proof}

\printbibliography

\end{document}